\documentclass[12pt]{iopart}

\usepackage{iopams}
\expandafter\let\csname equation*\endcsname\relax
\expandafter\let\csname endequation*\endcsname\relax
\usepackage{amsfonts,amsmath,amssymb}
\usepackage[english]{babel}
\usepackage[normalem]{ulem}

\usepackage{amsthm}
\usepackage{graphicx}
\usepackage{dsfont}
\usepackage{epsfig}
\usepackage{ulem}
\usepackage{tikz}
\usepackage{mathrsfs}
\usepackage{url}

\usepackage{bbold}             
\usepackage{bm}
\usetikzlibrary{shapes,arrows}
\usetikzlibrary{matrix}
\usetikzlibrary{calc,patterns,angles,quotes,arrows,automata}
\graphicspath{ {./imgs/}}

\usepackage[linesnumbered,noend]{algorithm2e}
\usepackage{algorithmic}
\RestyleAlgo{ruled}
\SetKwInOut{Input}{Input}
\SetKwInOut{Output}{Output}
\SetKwInOut{Parameters}{Parameters}
\SetKwIF{If}{ElseIf}{Else}{if}{:}{else if}{else}{endif}

\theoremstyle{definition}

\newtheorem{lemma}{Lemma}
\newtheorem{proposition}{Proposition}

\newtheorem{theorem}{Theorem}

\newtheorem{definition}{Definition}

\newtheorem{remark}{Remark}

\newcommand{\one}{\mathds{1}}
\newcommand{\ket}[1]{\lvert #1 \rangle}
\newcommand{\bra}[1]{\langle #1 \rvert}
\newcommand{\braket}[2]{\langle #1 \lvert #2 \rangle}
\newcommand{\ketbra}[2]{\lvert #1 \rangle \langle #2 \rvert}
\newcommand{\iu}{\mathrm{i}\mkern1mu}

\begin{document}

\title[Finding Quantum Codes via Riemannian Optimization]{Finding Quantum Codes via Riemannian Optimization}

\author{Miguel Casanova$^1$, Kentaro Ohki$^2$, Francesco Ticozzi$^{1,3}$}

\address{$^1$ Department of Information Engineering, University of Padova, Italy}
\address{$^2$ Graduate School of Informatics, Kyoto University, Japan}
\address{$^3$ QTech center, University of Padova, Italy}
\eads{\mailto{casanovame@dei.unipd.it}, \mailto{ohki@i.kyoto-u.ac.jp} and \mailto{ticozzi@dei.unipd.it}}

\begin{abstract}

We propose a novel optimization scheme designed to find
optimally correctable subspace codes for a known quantum noise channel. To each candidate subspace code we first associate a universal recovery map, as if the code was perfectly correctable, and aim to maximize a performance functional that combines a modified channel fidelity with a tuneable regularization term that promotes simpler codes.
With this choice optimization is performed only over the set of codes,
and not over the set of recovery operators. The set of codes of fixed dimension is parametrized as a complex-valued Stiefel manifold: the resulting non-convex optimization problem is then solved by gradient-based local algorithms. When perfectly correctable codes cannot be found, a second optimization routine is run on the recovery Kraus map, also parametrized in a suitable Stiefel manifold via Stinespring representation.
To test the approach, correctable codes are sought in different scenarios and compared to existing ones: three qubits subjected to
bit-flip errors (single and correlated), four qubits undergoing local amplitude
damping and five qubits subjected to local depolarizing channels.
Approximate codes are found and tested for the previous examples as well pure non-Markovian dephasing noise acting on a $7/2$ spin,
induced by a $1/2$ spin bath, and the noise of the first three qubits of IBM's \texttt{ibm\_kyoto} quantum computer. The fidelity results are competitive with existing iterative optimization algorithms, with respect to which we maintain a strong computational advantage, while obtaining simpler codes.
\end{abstract}


\submitto{Quantum Science and Technology}
\maketitle

\section{Introduction}
Quantum algorithms and their promised computational advantage rely, in the unitary-computation paradigm, on the assumption of being able to access and coherently manipulate logical information encoded in a physical quantum device for a sufficiently long time. Providing methods to increase the dauntingly short coherence time of most solid-state implementation by use of protected encodings \cite{lidarDecoherenceFreeSubspacesQuantum1998, knillTheoryQuantumError2000}, active error-correcting codes \cite{knillTheoryQuantumErrorcorrecting1997,gottesman}, or noise-suppressing controls \cite{violaDynamicalSuppressionDecoherence1998} has been the core business of the quantum error-correction (QEC) field \cite{lidarQuantumErrorCorrection2013}. Improving on existing experimental demonstrations \cite{cory1998experimental,chiaverini2004realization,schindler2011experimental,reed2012realization}, the development of these techniques is arguably a necessary step in order to surpass the current phase of noisy intermediate-scale quantum computers \cite{preskillQuantumComputingNISQ2018}.

In this work we focus on QEC codes: these are specified by of a set of quantum states (the {\em code})
faithfully representing some logical information, and a recovery
map that is able to correct the effect of certain classes of environmental noise on such states \cite{knillProtectedRealizationsQuantum2006,gottesman,lidarQuantumErrorCorrection2013}.

The foundational results in QEC provide conditions for exact error correction, and propose effective codes allowing for perfect avoidance or correction of certain classes of errors. The underlying assumption is that the considered error processes are assumed Markovian and represent the dominant part of the actual noise evolution. Even under such simplifying assumptions the task can be daunting: 
while the existence of suitable noiseless encodings of quantum information for a given noise model can be assessed by analyzing the spectral decomposition of the noisy dynamics \cite{choiMethodFindQuantum2006,ticozziFindingQuantumNoiseless}, it is known that finding perfectly correctable codes is an NP-Hard problem, even with perfectly known error models
\cite{blume-kohoutInformationPreservingStructures2010}.

Another route towards a systematic code design has been explored: assuming a given noise model, without overly simplifying assumptions, one can perform a direct numerical search of approximate quantum error correcting codes by optimizing a suitable fidelity functional over the possible encoding and recovery  \cite{reimpellIterativeOptimizationQuantum2005,kosutRobustQuantumError2008}.
The approach presents its own challenges, as the emerging optimization problems are in general non-convex and thus proving their convergence to global optima, i.e. the best available codes, is challenging. Furthermore, they often lead to codes that are challenging to interpret and, more crucially, implement: when phenomenological error models are considered, optima are degenerate and the numerical algorithm are unable to further select among them; when noise models reconstructed from experimental data are considered instead, limits in the accuracy and errors lead to codes with support on states with complex description.
The existing optimization-based approaches are summarized in more detail in Section \ref{sec:comparison}.

In this work, we re-consider this path towards better QEC, and propose a flexible tool that is able to:
(1) determine whether perfect (zero-error) subspace codes of a given dimension exist for a given error model;  and if they do not, (2) find optimal subspace codes and their corresponding recovery with minimal numerical effort, avoiding iterative optimization; (3) promote codes with simpler structure, easier to effectively implement.

The algorithm we propose is based on gradient-based
optimization on Riemannian manifolds. 
In particular, we represent 
quantum subspace codes of  fixed dimension as points in the complex-valued
Stiefel manifold. 
While more general subsystem codes have been proposed and realized \cite{knillTheoryQuantumError2000,viola2001experimental}, subspace codes are prevalent in the literature. In fact, one can show that for perfect QEC the choice is non-restrictive \cite{ticozziQuantumInformationEncoding2010}, they can be initialized via unitary control, and their description is easier. With this choice, we are able to reduce the complexity by avoiding the use of encoding maps (i.e., completely-positive trace-preserving mapping of logical information into physical states, whose output is the candidate code). 
This parametrization has also the advantage of returning a valid quantum
code with no need for
a final projection or renormalization step, unlike in e.g. \cite{kosutQuantumErrorCorrection2009}.

We initially avoid the computational burden of optimizing on the correction map, by  considering a time reversal or Petz recovery map on the code support
\cite{petzSufficiencyChannelsNeumann1988, barnumReversingQuantumDynamics2002, ticozziTimereversalSpacetimeHarmonic2010}.
The latter is known to be a perfect recovery map for correctable codes and, in general, a near-optimal correction map \cite{barnumReversingQuantumDynamics2002,ticozziTimereversalSpacetimeHarmonic2010}.

As for the cost function, we derive a quantity that measures the channel fidelity of the error-corrected evolution restricted to the code,
which is shown to be related to the usual channel fidelity.
Perfectly correctable codes are certified by saturation of this functional. 
A key feature of our method is the introduction of a tune-able regularization term in the optimization functional, aimed at promoting simpler codes, i.e. codes whose description is ``sparse''. This is obtained including a (small) term proportional to the $\ell_1$-norm of the code description in the functional. As we shall see, this modification allows us to address one of the main concerns regarding the complexity of the codes proposed by optimization algorithms, promoting simpler codes when the optima are multiple, and in general weighing the complexity of the code in the definition of the optimum. 

While the optimization problem is non-convex, the algorithm's limited run-time allows for multiple trials starting form different initial conditions in order to avoid local minima. 
As for realistic noise the existence of perfectly correctable codes is not expected, a single run of  a secondary optimization step, with the previously-obtained candidate as initial condition, can be then used to fine tune the correction operation when the universal one is not optimal.

The paper is structured as follows. We begin by giving a brief introduction to the key elements of QEC
theory of interest in this work in Section \ref{sec:qec}. Aiming to develop a general method based on the geometry of the problem, we shall adopt an abstract viewpoint inspired by \cite{knillProtectedRealizationsQuantum2006,ticozziQuantumInformationEncoding2010,blume-kohoutInformationPreservingStructures2010}.  We next proceed providing a derivation of the optimization problem
to be considered and the cost function to be used in Section \ref{sec:problem}. We shall prove that the basic functional we derive, without regularization term, is equivalent to a channel fidelity with input restricted to the code and can be use to certify that a code is perfectly correctable when its value is maximal. Section \ref{sec:numerical} introduces the Riemannnian
optimization framework and a numerical algorithm to solve the optimization problem. Some of the detailed calculations regarding the gradients are presented in the appendix to avoid overburdening the reader. Finally,
we show the effectiveness of the algorithm with a series of numerical examples in Section \ref{sec:examples}, comparing the proposed optimal codes with known codes and with random codes. Further details, and the comparison of regularized versus non-regularized codewords are reported in the appendix.

\section{Background and motivation}\label{sec:qec}

\subsection{QEC essentials}

In this work, we refer to {quantum noise} as the effect of a quantum environment on a quantum system of interest.
Assuming uncorrelated initial states, noise can be modelled in Schroedinger's
picture as a quantum Completly-Positive, Trace-Preserving (CPTP) map, also known as a quantum channel
\cite{nielsenQuantumComputationQuantum2012}.
Consider a finite dimensional system associated to a Hilbert space $\cal H,$ and let
$\mathcal{D}(\mathcal{H})$ be the set of density matrices, a CPTP map
$\mathcal{N}(\rho): \mathcal{D}(\mathcal{H}) \rightarrow \mathcal{D}(\mathcal{H})$,
can be written in the operator-sum representation \cite{krausStatesEffectsOperations1983}
$\mathcal{N}(\rho) = \sum_j N_j \rho N_j^\dagger$,
where $N_j$ are operators on $\cal H$ such that $\sum_j N_j^\dagger N_j = \mathds{1}$. When a CPTP map $\cal E$ admits a representation with kraus operators $E_j$ we write ${\cal E}\sim \{E_j\}.$

Quantum Error Correction (QEC) refers to a variety of techniques that aim to revert the effect of $\cal N$
on the encoded logical information on a quantum code $\mathcal{C} \subset \mathcal{D}(\mathcal{H}),$  a set of states designed to faithfully represent the logical information of interest \cite{knillProtectedRealizationsQuantum2006,ticozziQuantumInformationEncoding2010}.
In particular, we  focus on techniques designed to reverse the effect of a given {\em noise map} $\mathcal{N}(\rho)$, by implementing a CPTP {\em recovery map}
$\mathcal{R}(\rho).$
 We call a code $\cal C$ {\em correctable} for the noise $\cal N$ (or $\cal N$-correctable) if there exists a CPTP recovery map such that 
$\mathcal{R} \circ \mathcal{N}(\rho) = \rho,\quad \forall \rho \in \mathcal{C}$.

It is well known that classical approaches based on redundant encodings cannot be used directly in a quantum setting, due to the no-cloning theorem \cite{woottersSingleQuantumCannot1982}.
Another difference to classical information, is the fact that quantum errors are
continuous and of an infinite variety, instead of discrete combinations of bit flips
\cite{nielsenQuantumComputationQuantum2012}.
For these reasons, to properly design $\cal C$ and $\cal R$ one has to
heavily rely on the knowledge of the dynamical properties of the system
at hand and the noise that affects it.
In many applications, however, the dominant noise effects can be assumed on a phenomenological basis, or the free evolution of the system can be reconstructed from tomographic techniques \cite{maurodarianoQuantumTomography2003}.  For instance, a common assumption that can be made for systems
of qubits is that the noise is independent for each qubit
\cite{knillTheoryQuantumErrorcorrecting1997}, i.e. that in some representation each of the operation elements $N_j$ acts nontrivially only on a single qubit.  Indeed, these are the models
on which a substantial part of the early QEC theory has focused, and it
does not directly lend itself to generalization to physical systems that are not canonically
decomposable into qubits and that are subject to non-independent (correlated)
noises \cite{knillTheoryQuantumError2000}.
Some generalizations to systems of qudits of dimension $d$ do exist
\cite{ashikhminNonbinaryQuantumStabilizer2001, nadkarniQuantumErrorCorrection2021},
but it is required that $d$ be a prime number.
As for correlated noise, work addressing specific cases has been developed in e.g.
\cite{liEfficientQuantumError2011, liQuantumErrorCorrection2023, clemensQuantumErrorCorrection2004},
but they do not develop a general theoretical framework. Instead, they treat specific examples of
correlated noise models.

As the main aim of this work is to find correctable codes, we next recall what are the most general codes and, at the same time, why restricting to subspace codes is not restrictive. The information of interest,
originally associated to an  Hilbert space $\mathcal{H}_Q$, can be spread over a bigger
Hilbert space $\mathcal{H}_P$  through an adequate encoding, while remaining recoverable.
The most general codes are subsystem codes \cite{knillProtectedRealizationsQuantum2006},
i.e. any faithful encoding of a logical quantum system with Hilbert space $\mathcal{H}_Q$ can be expressed in terms of the following decomposition
\begin{equation}
    \label{eq:subsystem-enc}
    \mathcal{H}_P = (\mathcal{H}_{S} \otimes \mathcal{H}_F) \oplus \mathcal{H}_R,
\end{equation}
where $S$ is a realization of $Q$ as a subsystem of $P$, $F$ is called the
{cosubsystem} of $S$, and $R$ is the {remainder} subspace.
Under this decomposition, we have that
$\mathcal{C}$ has support only on $ \mathcal{H}_S \otimes \mathcal{H}_F$ \cite{ticozziQuantumInformationEncoding2010}.
It has been shown in \cite{kribsUnifiedGeneralizedApproach2005,ticozziQuantumInformationEncoding2010} that whenever there exists a correctable subsystem encoding, there is also a correctable subspace encoding, the latter being the first code model introduced in QEC theory.  A subspace encoding can be expressed
as the special case of decomposition (\ref{eq:subsystem-enc}) when
$\mathcal{H}_{S} \otimes \mathcal{H}_F = \mathcal{H}_{S}$. For this
reason, in the following we
are going to restrict our attention to subspace codes.

\subsection{Universal recovery maps from time reversal}
Let $\cal C$ be a subspace code, and $\Pi$ the orthogonal projection associated to $\mathcal{H}_S$ as above. Notice that each orthogonal projector $\Pi$ fully identify a subspace code as ${\cal C}_\Pi=\{\rho\in\mathcal{D}(\mathcal{H})|\Pi\rho\Pi=\rho\}$, so with some abuse of terminology in the following we shall sometimes refer to projectors as codes, and use ${\cal C}\sim \Pi$ to indicate the subspace code associated to the support of an orthogonal projection $\Pi$.
By definition, correctable codes satisfy $\mathcal{R} \circ \mathcal{N}(\rho) = \rho$
for every code word $\rho = \Pi \rho \Pi$. The existence of such a map ${\cal R}$ follows from the following condition,
given in Kraus representation,
\begin{equation}
   \label{eq:corrsubspace}
   \Pi N_j^\dagger N_k \Pi = \alpha_{jk} \Pi.
\end{equation}
where $[\alpha_{jk}]_{jk}$ forms a Hermitian matrix \cite{nielsenQuantumComputationQuantum2012}.
For a correctable subspace code and a given noise map, many different recovery operation may exist.
Nonetheless, a near-optimal one can be obtained as a special case of the time reversal map
$\mathcal{T}_{\rho,\mathcal{N}}.$ If ${\cal N} \sim \{ N_j\},$ then the latter is defined as $\mathcal{T}_{\rho,\mathcal{N}}\sim \{\rho^{1/2} N_j^\dagger (\mathcal{N}(\rho))^{-1/2}\}$,
also known as the Petz recovery map
\cite{petzSufficiencyChannelsNeumann1988, barnumReversingQuantumDynamics2002, ticozziTimereversalSpacetimeHarmonic2010}.
The general form of the time-reversal is dependent on the initial state, as it is constructed to ensure $\mathcal{T}_{\rho,\mathcal{N}}\circ\mathcal{N}(\rho)=\rho,$ but it can be
adapted to recover subspace codes by considering $\mathcal{T}_{\Pi/d_C,\mathcal{N}}$,
where $\Pi$ is the orthogonal projection with support on the code and $d_C$ is its dimension. Since
$(\Pi/d_C)^{1/2} N_j^\dagger \mathcal{N}(\Pi/d_C)^{-1/2} = \Pi^{1/2} N_j^\dagger \mathcal{N}(\Pi)^{-1/2},$ we have that $\mathcal{T}_{\Pi/d_C,\mathcal{N}}=\mathcal{T}_{\Pi,\mathcal{N}}.$
Therefore, we define the time-reversal (Petz) recovery map for the noise $\mathcal N$ and subspace code with projector $\Pi$ as
\begin{equation}
    \mathcal{R}_{\mathcal{N}, \Pi} \sim \{R_j\}, \quad
    R_j = \Pi^{1/2} N_j^\dagger \mathcal{N}(\Pi)^{-1/2} \ \forall j.
\end{equation}
Notice that since $\Pi$ is a projector, we have that $\Pi^{1/2} = \Pi$. Therefore,
the expression of the Kraus of operators of the recovery can be simplified to
$R_j = \Pi N_j^\dagger \mathcal{N}(\Pi)^{-1}$.
This map represent a perfect recovery for a correctable subspace code associated to $\Pi,$ as we recall next. The result is essentially known, see e.g. the discussion in \cite{barnumReversingQuantumDynamics2002}.
Nonetheless, a compact proof is included here for the sake of completeness.

\begin{proposition} $\mathcal{R}_{\mathcal{N}, \Pi}$ is a perfect recovery for a
$\cal N$-correctable code associated to $\Pi$.
\end{proposition}

\begin{proof} We want to show that the
application of this recovery implements a time-reversal of the noise map for all states
belonging to the code, which  is equivalent to $\mathcal{T}_{\mathcal{N}(\Pi),\mathcal{R}_{\mathcal{N}, \Pi}}={\cal N}$ for states such that $\Pi \rho \Pi = \rho$.
We have that the Kraus operators of $\mathcal{T}_{\mathcal{N}(\Pi),(\mathcal{R}_{\mathcal{N}, \Pi})}$ are
\begin{equation}
\begin{split}
    \mathcal{N}(\Pi)^{1/2} R_j^\dagger \Pi &=
    \mathcal{N}(\Pi)^{1/2} (\mathcal{N}(\Pi)^{-1/2} N_j \Pi) \Pi \\
    &= \Pi_{\mathcal{N}(\Pi)} N_j \Pi.
\end{split}
\end{equation}
Then, for $\Pi \rho \Pi$ one has that
\begin{equation}
    \begin{split}
        \mathcal{T}_{\mathcal{N}(\Pi),\mathcal{R}_{\mathcal{N}, \Pi}}(\Pi \rho \Pi)
        &= \sum_j \Pi_{\mathcal{N}(\Pi)} N_j \Pi \Pi \rho \Pi \Pi N_j^\dagger \Pi_{\mathcal{N}(\Pi)} \\
        &= \sum_j \Pi_{\mathcal{N}(\Pi)} N_j \Pi \rho \Pi N_j^\dagger \Pi_{\mathcal{N}(\Pi)},
    \end{split}
\end{equation}
where $\Pi_{{\cal N}(\Pi)} = {\cal N}(\Pi)^{-1} {\cal N}(\Pi)$. However, $\Pi N_j^\dagger \Pi_{\mathcal{N}(\Pi)}^\perp = 0$, therefore
\begin{equation}
    \Pi N_j^\dagger = \Pi N_j^\dagger (\Pi_{\mathcal{N}(\Pi)} + \Pi_{\mathcal{N}(\Pi)}^\perp) =
    \Pi N_j^\dagger \Pi_{\mathcal{N}(\Pi)},
\end{equation}
\begin{equation}
    \mathcal{T}_{\mathcal{N}(\Pi),\mathcal{R}_{\mathcal{N}, \Pi}}(\Pi \rho \Pi) =
    \sum_j  N_j \Pi \rho \Pi N_j^\dagger = \mathcal{N}(\Pi \rho \Pi).
\end{equation}
\end{proof}

Different ways to implement this map have been proposed see e.g. \cite{gilyenQuantumAlgorithmPetz2022}, while a 
continuous time version is proposed in \cite{kwonReversingLindbladDynamics2022}. In the following we will exploit the existence of a standardized recovery map in order to simplify the search for an effective error-correction protocol.

\section{Optimization problem}\label{sec:problem}

\subsection{Characterizing correctable subspace codes}

In this section we derive the cost function and the specific problem
we shall solve numerically by finding a suitable characterization of correctable subspace codes. The proposed characterization has the advantage
that it is independent of the initial state of the system, as it is based
only on the subspace code and the noise model.

Let $\mathcal{N}$ be a known noise channel defined as
\begin{equation}
    \mathcal{N}(\rho) = \sum_{j=1}^m N_j \rho N_j^\dagger,
    \label{eq:noise}
\end{equation}
where $\sum_{j=1}^m N_j^\dagger N_j = \mathds{1}_n$.
Consider the recovery map for the code identified by $\Pi$:
\begin{equation}
\begin{split}
    \mathcal{R}_{\mathcal{N}, \Pi}(\rho) &= \sum_{k=1}^m R_k \rho R_k^\dagger \\
    &= \sum_{k=1}^m \Pi N_k^\dagger {\cal N}(\Pi)^{-1/2} \rho {\cal N}(\Pi)^{-1/2} N_k \Pi.
\end{split}
\label{eq:petz}
\end{equation}
Recall that the Petz recovery map is a perfect recovery if $\Pi$ is correctable for
$\mathcal{N}$. Therefore, finding a correctable code for the noise $\cal N$ is equivalent  to find a projection matrix $\Pi$,
such that $\mathcal{R}_{\mathcal{N},\Pi}$ acts as the reversal of $\mathcal{N}$.
Relaxing the problem to a variational formulation, we want the composition of the noise and the recovery,
$\mathcal{R}_{{\cal N}, \Pi} \circ \mathcal{N}(\rho)$, to be as close as possible to the
identity map, when restricted to the subspace of the code. Next, we will derive a natural distance to measure how close the two maps are. Define:
\begin{equation}
    \mathcal{P}_\Pi(\rho) = \Pi \rho \Pi.
\end{equation}

In order to restrict the map to the code, we consider only input states such that
$\mathcal{P}_\Pi(\rho) = \rho$. For this reason, define the map
\begin{equation}\label{eq:Acheck}
    \mathcal{A}_\Pi(\rho) = \mathcal{R}_{{\cal N}, \Pi} \circ \mathcal{N} \circ \mathcal{P}_{\Pi} (\rho) =
    \sum_{j,k=1}^m \underbrace{\Pi N_k^\dagger \mathcal{N}(\Pi)^{-1/2} N_j \Pi}_{A_{jk} = R_k N_j \Pi} \rho
    \Pi N_j^\dagger \mathcal{N}(\Pi)^{-1/2} N_k \Pi.
\end{equation}

With this definition we have that ${\cal A}_\Pi(\rho) = \rho$ if and only if $\rho$
is in the code space and is perfectly corrected by the recovery map. On the other hand,
if $\rho$ has support othogonal to the code space, we have that ${\cal A}_\Pi(\rho) = 0$.
Clearly $\mathcal{A}_\Pi$ is not a trace preserving map, but a trace non-increasing map.
We next prove it is also self-adjoint with respect to the standard Hilbert-Schmidt product.

\begin{proposition}
    The map $\mathcal{A}_\Pi$
    is self-adjoint.
\end{proposition}
\begin{proof}
    The map $\mathcal{A}_\Pi$ can be written as in \eqref{eq:Acheck} as two-index operator sum with operators $A_{jk} =\Pi N_k^\dagger \mathcal{N}(\Pi)^{-1/2} N_j \Pi$. Even though these operators are not
    self-adjoint in general, they do satisfy the property that $A_{jk}^\dagger = A_{kj}$, i.e.
    taking adjoint corresponds to an exchange of the two indeces.
    Then, the adjoint operator sum can be written as
    \begin{equation}
    \begin{split}
        \mathcal{A}_\Pi^\dagger(\rho) &=
        \sum_{jk} A_{jk}^\dagger \rho A_{jk} =
        \sum_{j \leq k} (A_{jk}^\dagger \rho A_{jk} + A_{kj}^\dagger \rho A_{kj}) \\
        &= \sum_{j \leq k} (A_{kj} \rho A_{kj}^\dagger + A_{jk} \rho A_{jk}^\dagger)
        = \sum_{jk} A_{jk} \rho A_{jk}^\dagger \\
        &= \mathcal{A}_\Pi(\rho).
    \end{split}
    \end{equation}
\end{proof}

Next, in order to better exploit the linear structure of the maps, we consider
their vectorized representation.

The vectorization operation $\mathrm{vec}(X)$, defined as the stacking of the columns of matrix $X$, satisfies the property
$\mathrm{vec}(A X B) = (B^\top \otimes A) \mathrm{vec}(X)$. Then,
$\mathcal{A}_\Pi$ can be rewritten as
\begin{equation}
\begin{split}
    \mathrm{vec}(\mathcal{A}_\Pi (\rho)) &=
    \mathrm{vec}(\sum_{jk} A_{jk} \rho A_{jk}^\dagger) \\
    &=
    \underbrace{\sum_{jk} \overline{A_{jk}} \otimes A_{jk}}_{\hat{A}_\Pi} \mathrm{vec}(\rho),
\end{split}
\end{equation}
where
\begin{equation}
\begin{split}
    \hat{A}_\Pi &= \sum_{j,k=1}^m (\Pi^\top N_j^\top \mathcal{N}(\Pi)^{-\top/2} \overline{N_k} \Pi^\top)
    \otimes (\Pi N_k^\dagger \mathcal{N}(\Pi)^{-1/2} N_j \Pi) \\
    &= \sum_{jk} \overline{A_{jk}} \otimes A_{jk}.
\end{split}
\end{equation}

\begin{lemma}
    The operator $\hat{A}_\Pi$ is self-adjoint.
\end{lemma}
\begin{proof}
    Using the fact that $A_{jk}^\dagger = A_{kj}$, one can write
    \begin{equation}
        \begin{split}
            \hat{A} &=
            \sum_{j \leq k} (\overline{A_{jk}} \otimes A_{jk} + \overline{A_{kj}} \otimes A_{kj}) \\
            &=
            \sum_{j \leq k} (\overline{A_{jk}} \otimes A_{jk} + (\overline{A_{jk}} \otimes A_{jk})^\dagger) \\
            &= 2 \sum_{j \leq k} \mathrm{Sym}(\overline{A_{jk}} \otimes A_{jk}).
        \end{split}
    \end{equation}
\end{proof}

Let us now define the set of vectorized code words as $\mathcal{D}_\Pi,$ and study the structure that
$\hat{A}_\Pi$ acquires with respect to this subspace.

\begin{lemma}
    Define $\mathcal{D}_\Pi=\mathrm{span}\{\mathrm{vec}(\rho)|\rho=\Pi\rho\Pi, \rho\in{\cal D}(\mathcal{H})\}$ and consider the decomposition $\mathbb{C}^{n^2} = \mathcal{D}_\Pi \oplus \mathcal{D}_\Pi^\perp$.
    Let $T$ be a unitary such that for every state $\rho \in \mathcal{C}$,
    \begin{equation}
        T \rho T^\dagger = \left[\begin{array}{ c | c }
            \rho_\Pi & 0 \\ \hline
            0             & 0
        \end{array}\right].
    \end{equation}
    Then we have
\begin{equation}
        (\overline{T} \otimes T) \hat{A}_\Pi (\overline{T} \otimes T)^\dagger = \left[\begin{array}{ c | c }
            \tilde{A}_\Pi & 0 \\ \hline
            0             & 0
        \end{array}\right],
    \end{equation}
    where $\tilde{A}_\Pi$ is Hermitian
    \begin{equation}
        \tilde{A}_\Pi = \tilde{A}_\Pi^\dagger.
    \end{equation}
\end{lemma}
\begin{proof}
    Since $\Pi = \Pi^\dagger$, it is unitarily diagonalizable, i.e. there exists
    some $T$, such that $T T^\dagger = \mathds{1}$ and
    \begin{equation}
        T \Pi T^\dagger = \left[\begin{array}{ c | c }
            \mathds{1} & 0 \\ \hline
            0             & 0
        \end{array}\right].
    \end{equation}

    Since $\Pi \rho \Pi$ is in the support of $\Pi$, it is block-diagonalized by
    the same unitary $T$. Indeed, we have that
    \begin{equation}
        T \Pi \rho \Pi T^\dagger = T \Pi T^\dagger (T \rho T^\dagger) T \Pi T^\dagger =
        \left[\begin{array}{ c | c }
            \rho_\Pi & 0 \\ \hline
            0             & 0
        \end{array}\right].
    \end{equation}

    Similarly, for the operator $\hat{A}_\Pi$, since each of the $A_{ij}$ are in the support of $\Pi$, we have that
    \begin{equation}
        \begin{split}
            (\overline{T} \otimes T) \hat{A}_\Pi (\overline{T}^{\dagger} \otimes T^\dagger) =&
            \sum_{jk} \overline{(T A_{jk} T^\dagger)} \otimes (T A_{jk} T^\dagger) \\
            =&
            \sum_{jk} \overline{(T \Pi N_k^\dagger \mathcal{N}^{-1/2}(\Pi) N_j \Pi T^\dagger)}
            \otimes (T \Pi N_k^\dagger \mathcal{N}^{-1/2}(\Pi) N_j \Pi T^\dagger) \\
            =& \sum_{ij}
            \left[\begin{array}{ c | c }
                \overline{\tilde{A}_{jk}} & 0 \\ \hline
                0             & 0
            \end{array}\right] \otimes
            \left[\begin{array}{ c | c }
                \tilde{A}_{jk} & 0 \\ \hline
                0             & 0
            \end{array}\right] \\
            =&
            \left[\begin{array}{ c | c }
                \tilde{A}_\Pi & 0 \\ \hline
                0             & 0
            \end{array}\right].
        \end{split}
    \end{equation}
    
    Finally, Hermitianity of $\tilde{A}_\Pi$ is a direct consequence of the
    Hermitianity of $\hat{A}_\Pi$ and the block structure of
    $(\overline{T} \otimes T) \hat{A}_{\Pi} (\overline{T} \otimes T)^\dagger$.
\end{proof}
We are ready to state the main result of the section: a characterization of correctable codes in terms of the operator $\tilde{A}_\Pi$ introduce in the previous Lemma.
\begin{theorem}\label{thm:main}
    Let $d = \mathrm{rank}(\Pi)$, then, with $\tilde{A}_{\Pi}$ defined as above we have:
    \begin{enumerate}
        \item 
            $\tilde{A}_{\Pi}^\dagger \tilde{A}_\Pi \leq \mathds{1}_{d^2}$.
        \item $\mathcal{C} \sim \Pi$ is a correctable code iff
            $\tilde{A}_\Pi = \mathds{1}_{d^2}$.
    \end{enumerate}
\end{theorem}
\begin{proof}
    From the trace non-increasing condition one has that
    \begin{equation}
        \sum_{jk} A_{jk}^\dagger A_{jk} \leq \mathds{1}.
    \end{equation}
    Each of the summands $A_{jk}^\dagger A_{jk}$ is positive
    semi-definite. Therefore, all of their eigenvalues are non-negative. At the
    same time, to comply with the previous condition, it must also be the case
    that
    \begin{equation}
        \bra{\psi} (\sum_{jk} A_{jk}^\dagger A_{jk}) \ket{\psi} \leq 1.
    \end{equation}

    Putting this together with the non-negativity of the eigenvalues, one can
    conclude that
    \begin{equation}
        0 \leq \bra{\psi} (A_{jk}^\dagger A_{jk}) \ket{\psi} \leq 1,
    \end{equation}
    and combining the last two expressions, one has that
\begin{equation}
        0 \leq \bra{\psi} (\sum_{jk} A_{jk}^\dagger A_{jk}) \ket{\psi} \leq 1.
    \end{equation}
    Since this is true $\forall \ket{\psi}$, vectorizing and using the
    same transformation $T$ from the previous proposition, it is easy to see that
    \begin{equation}
        \tilde{A}_\Pi^\dagger \tilde{A}_\Pi \leq \mathds{1}_{d^2}.
    \end{equation}

    To prove the second part, it is easy to see that
    $\mathrm{rank}(\hat{A}_\Pi) \leq d^2$ due to the projector
    $\Pi$, which has $\mathrm{rank}(\Pi) = d$ and the fact that
    $\mathrm{rank}(A \otimes B) = \mathrm{rank}(A) \mathrm{rank}(B)$. This means that at
    least $n^2 - d^2$ of the eigenvalues of $\hat{A}_\Pi$ are zero.

    For the study of the rest of the eigenvalues, consider $\{\ket{\psi_k}\}$
    such that $\ketbra{\psi_k}{\psi_k} \in \mathcal{C}$.
    Then $\overline{\ket{\psi_a}} \otimes \ket{\psi_b} \notin \mathrm{ker}(\overline{\Pi} \otimes \Pi)$
    and $\mathrm{span}(\{\overline{\ket{\psi_a}} \otimes \ket{\psi_b}\}) \supseteq \mathrm{Im}(\hat{A}_\Pi)$.

    If the code is correctable, one has that
    ${\cal A}(\ketbra{\psi_a}{\psi_a}) = \ketbra{\psi_a}{\psi_a}$,
    \begin{equation}
        \Rightarrow \overline{\ket{\psi_a}} \otimes \ket{\psi_a} =
        (\sum_{jk} \overline{A_{jk}} \otimes A_{jk}) \overline{\ket{\psi_a}} \otimes \ket{\psi_a}.
    \end{equation}

    Therefore, there are $d$ eigenstates with eigenvalue $1$.

    Now consider the following two states
    \[
        \rho = (\lvert \alpha \rvert^2 \ketbra{\psi_a}{\psi_a}
        + \lvert \beta \rvert^2 \ketbra{\psi_b}{\psi_b}),
    \]
    \[
        \sigma = (\alpha \ket{\psi_a} + \beta \ket{\psi_b})
        (\overline{\alpha} \bra{\psi_a} + \overline{\beta} \bra{\psi_b}),
    \]
    with $\alpha, \beta \in \mathbb{C}$, $|\alpha|^2 + |\beta|^2 = 1$. It is
    easy to see that
    $\sigma = \rho + (\alpha \overline{\beta} \ketbra{\psi_a}{\psi_b} + \overline{\alpha} \beta \ketbra{\psi_b}{\psi_a})$.

    By linearity, one gets that $\mathcal{A}(\rho) = \rho$ and thus
    $\hat{A}_\Pi \mathrm{vec}(\rho) = \mathrm{vec}(\rho)$.

    Then
    $\mathcal{A}(\sigma) = \rho + \mathcal{A}(\alpha \overline{\beta} \ketbra{\psi_a}{\psi_b} + \overline{\alpha} \beta \ketbra{\psi_b}{\psi_a})$
    and
    \begin{equation*}
    \begin{split}
    \hat{A}_\Pi \mathrm{vec}(\sigma) &=
    \mathrm{vec}(\rho) + \hat{A}_\Pi \mathrm{vec}(\alpha \overline{\beta} \ketbra{\psi_a}{\psi_b} + \overline{\alpha} \beta \ketbra{\psi_b}{\psi_a}) \\
    &= \mathrm{vec}(\rho) + \hat{A}_\Pi (\alpha \overline{\beta} \overline{\ket{\psi_b}} \otimes \ket{\psi_a} + \overline{\alpha} \beta \overline{\ket{\psi_a}} \otimes \ket{\psi_b}).
    \end{split}
    \end{equation*}

    Hence, for $\mathcal{C}$ to be a correctable code, one would also need that
    $\hat{A}_\Pi \overline{\ket{\psi_b}} \otimes \ket{\psi_a} = \overline{\ket{\psi_b}} \otimes \ket{\psi_a}$,
    Since it must be true for any $a, b$, it results in $d^2 - d$ eigenstates having eigenvalue equal to $1$.

    Putting everything together, one sees that $\tilde{A}_\Pi$ has all of its
    eigenvalues equal to $1$, and being Hermitian, this also implies that it is
    equal to the identity. Proving the other direction of the implication is straightforward, since if $\tilde{A}_\Pi = \mathds{1}$, for any state in
    the code it would yield $\mathcal{A}_\Pi(\rho) = \rho$, due to the previous
    Lemma.
\end{proof}

\subsection{The code optimization problem}

In the light of Theorem \ref{thm:main}, it is natural to consider the following quality
index for the correctability of $\Pi$,
\begin{equation}
    J(\Pi) = \mathrm{tr}(\tilde{A}_\Pi) = \sum_k \mathrm{eigs}(\hat{A}_\Pi)_k.
    \label{eq:cost}
\end{equation}

The following proposition connects the correctability of a code to its optimality in terms of
$J(\Pi)$.

\begin{proposition}
    The function $J(\Pi)$ is non-negative and upper-bounded by the square
    of the rank of the projector $\Pi$,
    \begin{equation}
        0 \leq J(\Pi) \leq d^2.
    \end{equation}

    Furthermore, the code
    $\mathcal{C} \sim \Pi$ is a correctable code iff $J(\Pi) = d^2$.
\end{proposition}
\begin{proof}
    The $J(\Pi) \leq d^2$ part is a direct consequence of the first item of
    Theorem \ref{thm:main}, together with the fact that $\tilde{A}$ is Hermitian. The
    following proves the $J(\Pi) \geq 0$ part,
    \begin{equation}
        \begin{split}
            \mathrm{tr}(\hat{A}_\Pi) &= \sum_{kl}
            \mathrm{tr}(\overline{\Pi N_l^\dagger \mathcal{N}(\Pi)^{-1/2} N_k}) \\
            &\qquad \mathrm{tr}(\Pi N_l^\dagger \mathcal{N}(\Pi)^{-1/2} N_k) \\
            &= \sum_{kl} |\mathrm{tr}(\Pi N_l^\dagger \mathcal{N}(\Pi)^{-1/2} N_k \Pi)|^2 \geq 0.
        \end{split}
    \end{equation}

    Finally, using the results of the previous theorem, $J(\Pi) = d^2$ iff all the
    eigenvalues of $\tilde{A}$ are equal to $1$.
\end{proof}

Next, we try to relate this cost functional to existing QEC quantities.
To this aim, we define a {\em code-restricted} {fidelity}.
Let ${\cal E}: {\cal B}({\cal H}) \rightarrow {\cal B}({\cal H})$ be a quantum operation,
i.e. a trace non-increasing CP map from ${\cal B}({\cal H})$ to itself. Then let
${\cal C} \subset {\cal H}$ be a subspace code and
$\ket{\Omega_{\cal C}} = \sum_j \ket{j_{\cal C}} \otimes \ket{j_{\cal C}}/ \sqrt{d_{\cal C}}$,
where $\{\ket{j_{\cal C}}\}_{j=1}^{d_{\cal C}}$ is an orthonormal basis for the code and $d_{\cal C}$ its dimension. 

\begin{definition}
    We define the {\em code-restricted operation fidelity} (CRO) as 
\begin{equation}
F_O^{\cal C}({\cal E}) =
\bra{\Omega_{\cal C}} (\one \otimes {\cal E})(\ketbra{\Omega_{\cal C}}{\Omega_{\cal C}})
\ket{\Omega_{\cal C}}
\end{equation}
\end{definition}
Notice that this definition can be seen as a specialization of the channel fidelity:
let ${\cal E}: {\cal D}({\cal H}) \rightarrow {\cal D}({\cal H})$ be a quantum channel,
i.e. a CPTP map from ${\cal D}({\cal H})$ to itself, and
$\ket{\Omega} = \sum_j \ket{j} \otimes \ket{j}/ \sqrt{d}$,
where $\{\ket{j}\}_{j=1}^{d}$ is an orthonormal basis for the Hilbert space ${\cal H}$; then
the channel fidelity is defined as:
\begin{equation}
F_C({\cal E}) =
\bra{\Omega} (\one \otimes {\cal E})(\ketbra{\Omega}{\Omega})
\ket{\Omega}.
\end{equation}

\begin{lemma}
Let $\{E_j\}_j$ be the Kraus representation the quantum operation ${\cal E}$.
If ${\cal C} \subseteq {\rm supp}({\cal E})$, then
\begin{equation}
F_O^{\cal C}(\mathcal{E}) = \sum_{j} \frac{1}{d^2} | \mathrm{tr}(E_j) |^2 
\end{equation}
\end{lemma}
\begin{proof}
By direct calculation
\begin{equation*}
\begin{split}
F_O^{\cal C}(\mathcal{E}) =&
\bra{\Omega_{\cal C}} (\one \otimes {\cal E})[\ketbra{\Omega_{\cal C}}{\Omega_{\cal C}}] \ket{\Omega_{\cal C}} \\
=& \frac{1}{d_{\cal C}^2} \sum_{jklxy} \bra{x_{\cal C}} \otimes \bra{x_{\cal C}}
(\one \otimes E_j)
(\ket{k_{\cal C}} \otimes \ket{k_{\cal C}})
(\bra{l_{\cal C}} \otimes \bra{l_{\cal C}})
(\one \otimes E_j^\dagger)
\ket{y_{\cal C}} \otimes \ket{y_{\cal C}} \\
=& \frac{1}{d_{\cal C}^2} \sum_{jklxy} \delta_{xk} \delta_{ly} \bra{x_{\cal C}}
E_j \ket{k_{\cal C}} \bra{l_{\cal C}}
E_j^\dagger \ket{y_{\cal C}} \\
=& \frac{1}{d_{\cal C}^2} \sum_{jlx} \bra{x_{\cal C}}
E_j \ket{x_{\cal C}} \bra{l_{\cal C}}
E_j^\dagger \ket{l_{\cal C}} \\
=& \frac{1}{d_{\cal C}^2} \sum_{j} {\rm tr}(E_j)
\overline{{\rm tr}(E_j)} \\
=& \sum_{j} \frac{1}{d_{\cal C}^2} | \mathrm{tr}(E_j) |^2
\end{split}
\end{equation*}
\end{proof}

\begin{proposition}
    The cost function $J_d(\Pi)$ with $d$ fixed is proportional to the CRO
    fidelity of $\mathcal{R}_{{\cal N}, \Pi} \circ \mathcal{N} \circ {\cal P}_\Pi$.
    \begin{equation}
        J_d(\Pi) = d^2 F_O^{\cal C}(\mathcal{R}_{\mathcal{N}, \Pi} \circ \mathcal{N} \circ {\cal P}_\Pi).
    \end{equation}
\end{proposition}
\begin{proof}
    By direct calculation,
    \begin{equation*}
        \begin{split}
            J_d(\Pi) &= \mathrm{tr}(\tilde{A})
            = \sum_{jk} \mathrm{tr}(\overline{A_{jk}} \otimes A_{jk}) \\
            &= \sum_{jk} \overline{\mathrm{tr}(A_{jk})} \mathrm{tr}(A_{jk})
            = \sum_{jk} | \mathrm{tr}(A_{jk}) |^2 \\
            &= \sum_{jk} | \mathrm{tr}(R_k N_j \Pi) |^2
            = d^2 \sum_{jk} \frac{1}{d^2} | \mathrm{tr}(R_k N_j \Pi) |^2 \\
            &= d^2 F_O^{\cal C}(\mathcal{R}_{{\cal N}, \Pi} \circ \mathcal{N} \circ {\cal P}_\Pi).
        \end{split}
    \end{equation*}
\end{proof}

Consider the following
optimization problem with fixed $d = {\rm rank}(\Pi)$,
\begin{equation}
    \Pi_C = \underset{\Pi}{\mathrm{argmax}} J_d(\Pi).
\end{equation}

Since $d$ is fixed, one could also decompose the projector onto the code as
$\Pi = U U^\dagger$, where $U \in \mathbb{C}^{n \times d}$ and $U^\dagger U = \mathds{1}_d$.
This decomposition leads one to perform optimization on the complex valued
Stiefel manifold $V_d(\mathbb{C}^n)$, since $U \in V_d(\mathbb{C}^n)$:
\begin{equation}
    \Pi_U = \underset{U \in V_d(\mathbb{C}^n)}{\mathrm{argmax}} J_d(\Pi)|_{\Pi = U U^\dagger}.
    \label{eq:optimprob}
\end{equation}

In the following section we shall focus on the Stiefel manifold $V_d(\mathbb{C}^n),$ its definition, its parametrization,
and methods to solve this optimization problem. Furthermore, in what follows we will implicitly use the notation
$J_d(U) = J_d(\Pi)|_{\Pi = U U^\dagger}$, and $J(U) = J_2(U)$.

\subsection{Regularization}

As we shall see in the examples, if a correctable code exists, it is typically
not unique. In order to obtain codes that are simpler to implement and
interpret, we introduce $\ell_1$ regularization terms in the optimization
functional \cite{murphyMachineLearningProbabilistic2013}. Intuitively, 
$\ell_1$ regularization is based on the geometry of $\ell_1$ balls
$B_{\ell_1}(r) = \{U \in V_d(\mathbb{C}^n)\ | \ \|U\|_1 \leq r\}.$ Their boundaries are not
smooth and have extreme points, corresponding to sparse elements in the set, with many zero entries.
For this reason,
addition of a (small) term proportional to $\|U\|_1$ to the cost function promotes
sparsity of the solutions of the optimization problem by selecting, on a given level set of a cost function,  extreme points of balls with smaller radius. In the following, a
regularization terms is added on the norm of the point $U$ of the Stiefel
manifold, in order to increase the sparsity of the code, encouraging the
convergence towards codes simpler distributions of amplitudes (with respect to a
suitable basis). The final objective function is the following,
\begin{equation}
    \label{eq:regcost}
    J_d^{\mathrm{reg}}(U) = J_d(U) + \lambda \|U\|_1,
\end{equation}
where $\|U\|_1 = \sum_{jk} |U_{jk}|$ and where $\lambda \geq 0$ is the regularization hyperparameter.

Notice, however, that the addition of regularization may have a detrimental effect if no sparse correctable
code exists and the hyperparameter $\lambda$ is too large. In general, there is a trade-off between sparsity
and performance. The optimal choice for the value of $\lambda$ is dependent on the optimization problem, and more specifically on the noise model. Therefore, careful tuning is necessary.
An example of this effect is given in Section \ref{sec:examples}.

\subsection{The Recovery Optimization Problem}

As mentioned above, when the noise is not perfectly correctable
the Petz recovery map is not optimal, but can be shown to be near-optimal under suitable assumptions \cite{barnumReversingQuantumDynamics2002}. In order to improve the fidelity preservation,
an additional optimization step may be performed: assuming the code ${\cal C} \sim \Pi$ to be fixed, one can then optimize
over the recovery map, using ${\cal R}_{{\cal N}, \Pi}$ as starting point. In fact, Riemannian optimization
on the Stiefel manifold may also be used in this step.

In order to do so, first parametrize the recovery map using Stinespring's isometry
\cite{wolfQuantumChannelsOperations2012}, i.e.
\begin{equation}
\hat{R} = \sum_{j=0}^{r-1} \ket{j} \otimes R_{j+1} = \begin{pmatrix}
R_1 \\
R_2 \\
\vdots \\
R_r
\end{pmatrix},
\end{equation}
where ${\cal R} \sim \{R_j\}$ is the Kraus representation of the recovery map. One can then verify that
$\sum_{j=1}^r R_j^\dagger R_j = \hat{R}^\dagger \hat{R} = \one_n$, i.e. $\hat{R}$ is a point of the Stiefel
manifold $V_n(\mathbb{C}^{r n})$.

Finally, by fixing $r$ and with a slight abuse of notation, we can adapt the cost function \eqref{eq:cost} as
\begin{equation}
\label{eq:cost_recovery}
J_{d,r}(\hat{R}) = \sum_{jk} | {\rm tr}(R_j N_k \Pi) |^2,
\end{equation}
and solve the following problem
\begin{equation}
    \hat{R}_U = \underset{\hat{R} \in V_n(\mathbb{C}^{r n})}{\mathrm{argmax}} J(\hat R).
    \label{eq:optimprob2}
\end{equation}
In what follows, we will set $r$ equal to the number of Kraus operator in the chosen noise representation, which is also the number of operators in thecPetz recovery.

\subsection{Existing optimization-based approaches to QEC: a brief review and comparison}\label{sec:comparison}

The first work using optimization algorithms in the search for QEC codes is \cite{reimpellIterativeOptimizationQuantum2005}, where an adapted version of the
power method algorithm has been proposed to maximize the
channel fidelity. The algorithm iterates alternate optimization of the encoding and the decoding, parametrized as CP maps, and
normalizes to obtain valid quantum channels at each iteration. This algorithm
is guaranteed to converge to local maxima if the Kraus rank of the encoding and
decoding channels is allowed to be large enough.

Then, a line of work focused on the use of semi-definite programming (SDP) to
optimize the Choi representation of the recovery channel can be traced back to
\cite{yamamotoSuboptimalQuantumerrorcorrectingProcedure2005}, in which they
assume fixed encodings and aim to maximize the minimum fidelity. A similar
scheme is proposed in \cite{fletcherOptimumQuantumError2007}, where they
substitute the minimum fidelity by the entanglement fidelity between the
initial and the recovered states, in order to simplify the problem. However,
this simplification has the downside that it makes the problem dependent on the
initial state of the system, and thus less general.

 In \cite{kosutRobustQuantumError2008}, an approximate recovery is used, and unitary encoding is assumed (by the inclusion of an
ancillary system). Then, they find the optimal encoding by optimizing the
distance, induced by the Frobenius norm, between the encoding-noise-recovery
map and the identity map. The resulting optimization problem is not convex:
in order to use SDP, they relax the unitarity condition and apply the singular
value decomposition to the result of the relaxed optimization problem. A
similar approach is followed in \cite{kosutQuantumErrorCorrection2009}, but the
optimization is performed on both encoding and decoding by iteratively fixing
one of the two in order to minize the channel fidelity, as in
\cite{reimpellIterativeOptimizationQuantum2005}.

All the previously mentioned works require a good model of the noise. More recent work
try to eliminate this need by resorting to data-driven methods. For example, in \cite{nautrupOptimizingQuantumError2019} a reinforcement learning scheme is proposed to find optimum codes. Here
they restrict the class of quantum codes to surface codes, which is a class of
quantum codes with practical implementations, and the noise is then assumed to
be a black box, together with the recovery operation. Then, they find the
optimal surface code that minimizes the error probability. Finally, in
\cite{zorattiImprovingSpeedVariational2023} instead of relying on classical
optimization algorithms, they use a variational quantum algorithm to find the
optimal code and recovery. However, the proposed algorithm requires the
specification an initial state, like in \cite{fletcherOptimumQuantumError2007}.

With respect to ours, none of the above methods has a way to promote codes with simple descriptions, and all the algorithms that are able to optimize both the encoding and the recovery rely on iterative methods. Our method also presents an advantage in the parameter scaling: if the system Hilbert space ${\cal H}$ is of
dimension $n$, and the subspace code ${\cal C}$ of dimension $d$, existing optimization methods employ representations of dimension $n \times n$ or 
$n^2 \times n^2,$ if Choi's representation is used \cite{yamamotoSuboptimalQuantumerrorcorrectingProcedure2005}. In our method, the optimization variables are represented by
matrices of dimension $n \times d$. In terms of computation time, the lack of iteration and the simpler parametrization ensure a strong advantage, especially when perfectly correctable codes do not exist: as an example, in the amplitude damping case discussed in Section \ref{sec:examples}, with the full error model our Riemannian methods provide a result in 487 seconds on a commercial laptop, while the iterative SDP takes approximately one day to converge.

\section{A solution using Riemannian optimization}
\label{sec:numerical}
A numerical approach towards the solution of the optimization problems we formulated may be
pursued using geometric ideas: a subspace can be associated to a projector, which in turn can be
expressed in terms of an orthonormal k-frame. Hence, the set of subspace codes
is encoded in the complex-valued Stiefel manifold \cite{satoRiemannianOptimizationIts2021}.  This observation opens the
door to the use of new optimization algorithms for the search of error correcting
codes. The Stiefel manifold is a Riemannian manifold, therefore one
can pursue unconstrained optimization on the full manifold, which naturally embeds the  constraints on the optimization parameters, instead of
performing constrained Euclidean optimization on $\mathbb{C}^{n \times d}$. We shall focus on gradient-based methods, where the gradient now has to be computed with respect to the Riemannian structure.

In the following we recall some basic facts regarding Stiefel manifolds and related optimization algorithms.

\subsection{Stiefel manifold and Riemannian gradient}
The complex valued Stiefel manifold is defined as
\begin{equation}
V_d(\mathbb{C}^n) = \{U \in \mathbb{C}^{n \times d} \ | \ U^\dagger U = \mathds{1}_d\},
\end{equation}
as a submanifold of $\mathbb{C}^{n \times d}$, where $n \geq d$.  It is easy to show
that $U U^\dagger = \Pi$ is a projector. This allows us to perform optimization directly
on $V_d(\mathbb{C}^n),$ hence automatically preserving the required geometrical properties
of the code.

For any $U, V \in V_d(\mathbb{C}^n)$, there exists a smooth function
$X : [0,1] \rightarrow V_d(\mathbb{C}^n)$ such that $X(0) = U$ and $X(1) = V$.
Considering the constraint $X(t)^\dagger X(t) = \mathds{1}_d$, we have that
\begin{equation*}
    \frac{d}{dt} (X(t)^\dagger X(t)) = \frac{d X(t)^\dagger}{dt} X(t) +
    X(t)^\dagger \frac{d X(t)}{dt} = 0_{d,d},
\end{equation*}
which defines the tangent vector space at $U \in V_d(\mathbb{C}^n)$,
\begin{equation}
    \mathcal{T}_U V_d(\mathbb{C}^n) = \{X \in \mathbb{C}^{n \times d} \vert
    X^\dagger U + U^\dagger X = 0_{d,d} \}.
\end{equation}

Clearly, in the case when $d = 1$,
$\mathcal{T}_U V_1(\mathbb{C}^n)$ is a set of orthogonal vectors (with respect to
the Euclidean inner product) to $U$. Consider the following maps,
\begin{equation}
    \pi_U(X) = X - U \mathrm{Sym}(U^\dagger X),
\end{equation}
\begin{equation}
    \pi_U^\perp(X) = U \mathrm{Sym}(U^\dagger X),
\end{equation}
for $X \in \mathbb{C}^{n \times d}$, where $\mathrm{Sym}(Z) = (Z + Z^\dagger)/2$
for $Z \in \mathbb{C}^{d \times d}$. From the
definition, $X = \pi_U(X) + \pi_U^\perp(X)$ for any $X \in \mathbb{C}^{n \times d}$
and $\pi_U(U) = 0_{n,d}$.
It is also easy to see that $\pi_U(X)^\dagger U + U^\dagger \pi_U(X) = 0_{n,d}$ for all
$X \in \mathbb{C}^{n \times d}$, i.e. that $\pi_U(X) \in {\cal T}_U V_d(\mathbb{C}^n)$.
One can then verify that $\pi_U(X) = X\ \forall X \in {\cal T}_U V_d(\mathbb{C}^n)$,
and that $\pi_U \circ \pi_U = \pi_U$: hence,
$\pi_U$ is a projection onto $\mathcal{T}_U V_d(\mathbb{C}^n)$.
We also have that
\[
    g_U(\pi_U(X), \pi_U^\perp(X)) = 0 \ \forall X \in \mathbb{C}^{n \times d},
\]
where $g_U$ is an inner product at $U \in V_d(\mathbb{C}^n)$ defined as
\begin{equation}
    g_U(X,Y) = \mathrm{tr}_{\mathbb{R}}(X^\dagger (\mathds{1}_n - \frac{1}{2} U U^\dagger) Y),
    \label{eq:caninprod}
\end{equation}
for $X, Y \in \mathbb{C}^{n \times d}$. Such $g_U$ is called the {\em canonical inner product}, see Section 2.4 of
\cite{edelmanGeometryAlgorithmsOrthogonality1998}. Note that $U \in V_d(\mathbb{C}^n)$ and
$X \in \mathcal{T}_U V_d(\mathbb{C}^n)$ are orthogonal with respect to the canonical
inner product (\ref{eq:caninprod}), but this does not imply that
$U^\dagger X = 0_{d,d}$ in general.

In order to adapt gradient-based optimization methods to the Stiefel manifold,
we need to use the Riemannian gradient and to find an update rule $R_{U}$
whose codomain is $V_d(\mathbb{C}^n)$.

The Riemannian gradient  of $J_d$ at $U \in V_d(\mathbb{C}^n)$
is given by \cite{satoRiemannianOptimizationIts2021}:
\begin{equation}
    \mathrm{grad} J_d(U) = \nabla_U J_d(U) - \pi_U^\perp(\nabla_U J_d(U)),
\end{equation}
where $\nabla_U J_d(U)$ is the usual Euclidean gradient. The explicit computation of
$\nabla_U J_d(U)$ is a key result of our work, but being fairly long and technical has been deferred to  \ref{chp:gradient}. It is easy to
verify that $\mathrm{grad} J_d(U) \in \mathcal{T}_U V_d(\mathbb{C}^n)$.
Therefore, the solution of the ordinary differential equation
\begin{equation}
    \frac{d}{dt} U(t) = \mathrm{grad} J_d(U(t)),
\end{equation}
with $U(0) \in V_d(\mathbb{C}^n)$ always lies on $V_d(\mathbb{C}^n)$ and if
the steady state solution exists, it is a local optimum of $J_d$.
To describe these dynamics of the Stiefel manifold, it is useful to introduce the
exponential map. The exponential map on $V_d(\mathbb{C}^n)$ with respect to the
canonical inner product (\ref{eq:caninprod}) defined as follows.

\begin{definition}[Exponential map]
    A map
    \begin{equation}
        \exp_U : \mathcal{T}_U V_d(\mathbb{C}^n) \rightarrow V_d(\mathbb{C}^n)
    \end{equation}
    is called an exponential map with respect to the canonical inner product (\ref{eq:caninprod})
    if
    \begin{equation}
        \exp_U(0_{n,d}) = U \text{ and } \frac{d}{dt} \exp_U(t X) \rvert_{t=0} = X,
    \end{equation}
    for $U \in V_d(\mathbb{C}^n)$ and $X \in \mathcal{T}_U V_d(\mathbb{C}^n)$.
\end{definition}

\subsection{Numerical algorithms}

The exponential map for the complex Stiefel manifold can be derived explicitly.
One of the computationally viable ways to obtain it
\cite{zimmermannMatrixAlgebraicAlgorithmRiemannian2017}
is recalled in Lemma \ref{thm:exponential}.

\begin{lemma}
\label{thm:exponential}
    For a given $U \in V_d(\mathbb{C}^n)$, consider $X \in \mathcal{T}_U V_d(\mathbb{C}^n)$.
    Define $A = U^\dagger X \in \mathbb{C}^{d \times d}$ and consider the
    QR decomposition of $(\mathds{1}_n - U U^\dagger) X$, which factorizes it into a product of $Q \in V_d(\mathbb{C}^n)$ and an upper triangular matrix
    $R \in \mathbb{C}^{d \times d}$.
    Then, the following map is an exponential map with respect to the inner product
    (\ref{eq:caninprod}).
    \begin{equation}
        \exp_U^c(X) =
        \begin{bmatrix}
            U & Q
        \end{bmatrix}
        \exp\left(
        \begin{matrix}
            A & -R^\dagger \\
            R & 0_{d,d}
        \end{matrix}
        \right)
        \begin{bmatrix}
            \one_d \\
            0_{d,d}
        \end{bmatrix},
    \end{equation}
    where $\exp$ represents the matrix exponential function.
\end{lemma}

However, numerical computation of a matrix exponential may be too slow for an iterative
optimization algorithm and, for this reason, other methods have been proposed. Roughly speaking,
these methods allow an approximate exponential to leave $V_d(\mathbb{C}^n)$, subsequently projecting  it back onto
$V_d(\mathbb{C}^n)$ using retraction map. The QR decomposition can be used to obtain
a viable retraction. The (thin, or reduced) QR decomposition of a rectangular $X \in \mathbb{C}^{n \times d}$
is $X = Q R$, where $Q \in V_d(\mathbb{C}^n)$ and $R$ is an upper triangular matrix \cite{boydIntroductionAppliedLinear2018}.
We thus need to extract the unitary factor of the QR decomposition, we define $\mathcal{Q} : \mathbb{C}^{n \times d} \rightarrow V_d(\mathbb{C}^n)$ to
be the map that returns the $Q$ factor.
At this point, we have all the ingredients necessary to define Algorithm \ref{alg:graddesc}.
In the following, $R_{U_k}(t d_k)$ refers to either the exponential map $\mathrm{exp}_U(t d_k)$
or a retraction, such as $\mathcal{Q}(U_k + t d_k)$.

\begin{algorithm}
    \caption{Gradient descent}
    \label{alg:graddesc}
    \begin{algorithmic}[1]
        \REQUIRE $U_0 \in V_d(\mathbb{C}^n)$
        \ENSURE $U_f \in V_d(\mathbb{C}^n)$
        
        \STATE $d_0 \gets \mathrm{grad} J_d(U_0)$
        \STATE $k \gets 0$
        
        \WHILE{\NOT Stopping conditions}
            \STATE $\Phi_k(t) \gets J_d(R_{U_k}(t d_k))$
            \STATE $t_k \gets \mathrm{argmax}_{t \in [0, \infty)} \Phi_k(t)$
            \STATE $U_{k+1} \gets R_{U_k}(t_k d_k)$,
            \STATE $d_{k+1} \gets \mathrm{grad} J_d(U_{k+1})$
            \STATE $k \gets k + 1$
        \ENDWHILE

        \STATE $U_f \gets U_k$
    \end{algorithmic}
\end{algorithm}

Since $\Pi_U = U U^\dagger = U R R^\dagger U^\dagger$
for any unitary $R \in \mathbb{C}^{d \times d}$, the optimal solution is not a
unique $U$, if it exists. In other words, the minima of $J_d(U)$ are degenerate,
and each local optimum $U$ belongs to a class of optima characterized by $\Pi_U$.
This degeneracy is broken once we consider the regularization introduced in (\ref{eq:regcost}).
However, also $J_d^{\mathrm{reg}}(U)$ may have multiple local optima
on $V_d(\mathbb{C}^n)$. For this reason, we run algorithm \ref{alg:graddesc}
with several initial values (multi-start optimization), and then select the best of the given results.
 More details on gradient flows can be found in \cite{helmkeOptimizationDynamicalSystems1994}.

Notice also that Algorithm \ref{alg:graddesc} involves the following optimization problem
$t_k \gets \mathrm{argmax}_{t \in [0, \infty)} \Phi_k(t)$. The variable $t_k$ is interpreted
as a step length, and the objective is to use the step length that brings the
estimate to the minimum of the curve $\Phi_k(t)$, which follows the descent direction.
This condition can be relaxed to the following,
\begin{equation}
    \label{eq:armijo}
    \Phi_k(t) \leq \Phi_k(0) + c_1 t_k g_{U_k}(\mathrm{grad} J_d(U_k), d_k),
\end{equation}
which is known as Armijo condition.
In order to find a step size that
satisfies the Armijo condition the backtracking algorithm is used.

\begin{algorithm}[!htpb]
    \caption{Backtracking}\label{alg:backtracking}
    \begin{algorithmic}
        \REQUIRE $U_k \in V_d(\mathbb{C}^n)$, $c_1$, $t_0$
        \ENSURE $t_k$

        \STATE $t \gets t_0$

        \STATE $d_k \gets -\mathrm{grad} J_d(U_k)$

        \WHILE{\NOT Armijo condition \eqref{eq:armijo}}
            \STATE $t \gets \tau t$
        \ENDWHILE

        \STATE $t_k \gets t$
    \end{algorithmic}
\end{algorithm}

Additionally, the two-point method developed in \cite{barzilaiTwoPointStepSize1988}
can be used to speed up the calculation of the step size by substituting the fixed
initial estimate of the backtracking algorithm by the one given by this method.
Given the current and previous estimates of the optimization algorithm, $U_k$,
$U_{k-1}$, the two-point method provides the following initial estimate for the
step size $t$. Let $\Delta U_k = U_k - U_{k-1}$ be the difference between the
last two iterations, and $\Delta G_k = G_k - G_{k-1}$ the difference
between their gradients, i.e. $G_k = \nabla_{U} J_d(U) \rvert_{U=U_k}$,
then the two-point step size is
\begin{equation}
    t_k = \frac{\langle \Delta U_k, \Delta G_k \rangle}{\langle \Delta G_k, \Delta G_k \rangle},
\end{equation}
where $\langle A, B \rangle = \mathrm{tr}(A^\dagger B)$.
Finally, note that due to numerical error the iteration point $U_k$
may cause the flow $\{U(t)\}_{t \geq 0}$ to exit
$V_d(\mathbb{C}^n)$. Usually the QR decomposition is employed to project
the flow back onto the Stiefel manifold, i.e.
$U_k \gets \mathcal{Q}(U_k)$. Another option is to use
$U_k \gets U_k(U_k^\dagger U_k)^{-1/2}$, if the inverse exists.

\section{Examples and applications}\label{sec:examples}

In this section we apply our method to a series of test-bed systems and noise models, and compare its
performance to the iterative SDP method described in \cite{lidarQuantumErrorCorrection2013}, as well as to relevant known quantum codes from the QEC literature.
The code and the data used in the following can be found in \cite{casanovaCodeDataFinding2024}.

\subsection{Optimal Codes for Bit-flip Noise Models}

In a system of qubits, a bit-flip noise is a quantum channel that rotates the
state of a random qubit around the $X$ axis of the Bloch sphere, with
probability $p$. This rotation is represented by the first Pauli matrix
\begin{equation}
    \sigma_x = \begin{bmatrix}
        0 & 1 \\
        1 & 0
    \end{bmatrix},
\end{equation}
which maps $\ket{0} \mapsto \ket{1}$ and $\ket{1} \mapsto \ket{0}$.
The Kraus operators for this noise acting
on a single qubit are
\begin{equation}
\label{eq:bitflip1q}
        \mathcal{N} \sim \{\sqrt{1 - p} \mathds{1}, \sqrt{p} \sigma_x\}.
\end{equation}
In this section we consider quantum codes consisting of three qubits,
since this is the smallest amount of qubits necessary to correct a single
bit-flip. However, in order to do so, we need to extend the noise model
\eqref{eq:bitflip1q} to the full dimension of the system in consideration.
We first consider a model extends \eqref{eq:bitflip1q} in a way such that
only a single bit-flip may occur at each iteration of the noise channel,
which we call independent noise model. Then we include a correlated
term that flips two qubits simultaneously. We use this as a toy model
to study the effect of such terms on the optimality of quantum codes.
Then, finally, we consider a model that includes all possible combinations
of $\one$ and $\sigma_x$ on all three qubits as a collection of independent
identically distributed binary experiments, which we call full noise model.

As benchmark, the repetition code \cite{nielsenQuantumComputationQuantum2012, lidarQuantumErrorCorrection2013}
for bit-flip errors is used in this section. The
basis for this code is given in \eqref{eq:shorcode}.

\subsubsection{Independent noise model on 3 qubits with perfect correction}

The independent bit-flip noise model for three qubits is given by
\begin{equation}
\label{eq:indybitflip}
    \mathcal{N} \sim \{\sqrt{1 - p} \mathds{1}^{\otimes 3}\} \cup
    \{\sqrt{\frac{p}{3}} \bigotimes_{j=1}^3 \sigma_x^{\delta_{jk}}\}_{k=1}^3.
\end{equation}

After several runs the Algorithm \ref{alg:graddesc}, it has been observed that the cost function $J_3(U)$,
under this noise channel, has multiple optima, which are all global. Both our method and the iterative
SDP method from \cite{lidarQuantumErrorCorrection2013} are able to find codes and recovery maps that
correct this noise with CRO fidelity equal to $1$. In fact, they always finds perfectly correctable codes
for this noise.

\subsubsection{Code dependency on dominant noise terms}

Before proceeding to the full noise model, let us consider a toy model where a  single correlated term
that flips two qubits simultaneously is introduced. As such, we modify the Kraus
operators of the noise channel as follows
\begin{equation}
\label{eq:toycorrbitflip}
    \mathcal{N} \sim \{\sqrt{1 - p} \mathds{1}^{\otimes 3}\} \cup
    \{\sqrt{\frac{p (1 - q)}{3}} \bigotimes_{j=1}^3 \sigma_x^{\delta_{jk}}\}_{k=1}^3 \cup
    \{\sqrt{p q} \mathds{1} \otimes \sigma_x \otimes \sigma_x\}.
\end{equation}
Here a correlated term $\sigma_x \otimes \sigma_x$ has been introduced, allowing for
a simultaneous bit-flip in the second and third qubits. After running
Algorithm \ref{alg:graddesc} with multiple initial points, with and without
the correlated term, we observed that the addition of this term does not alter
the structure of the degenerate global optima of the cost function, as long as the coefficient
of the correlated operation element is less than that of the uncorrelated
operation elements, i.e. $\sqrt{p (1 - q)/3} > \sqrt{p q}$ or equivalently
$q < 0.25$. If this condition is satisfied, the addition of correlation only lowers
the value of the optimum, but the optimum points do not change. However, if $q > 0.25$,
then the optimal codes are those that target the correlated noise, instead. This transition corresponds to the point in which the correlated error probability becomes dominant with respect to the single error terms. It suggests that the optimal codes are determined indeed by the dominant effects in the error model.
This effect can be seen in Figures \ref{fig:bitflip_corr} and \ref{fig:bitflip_corr_coeffs}.

\begin{figure}[!htpb]
    \centering
    \includegraphics[width=0.9\columnwidth]{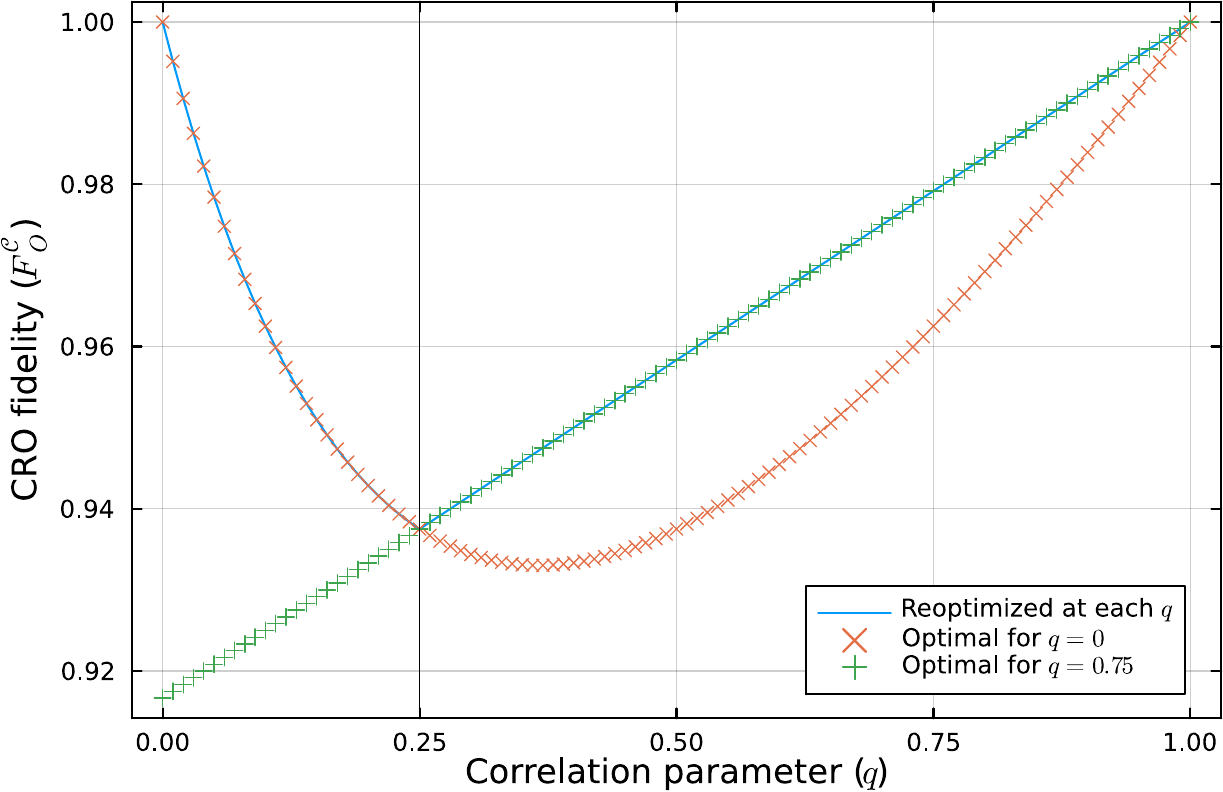}
    \caption{CRO fidelity for the correlated bit flip noise and different correlation strengths $q$. The red x marks and green crosses correspond to fixed codes optimized at $q=0$ and $q=0.75$, respectively, whereas the blue solid line correspond to codes that were reoptimized for each value of $q$. The recovery used in all three cases was the Petz recovery map, adjusted for each value of $q$.}
    \label{fig:bitflip_corr}
\end{figure}

\begin{figure}[!htpb]
    \centering
    \includegraphics[width=0.9\columnwidth]{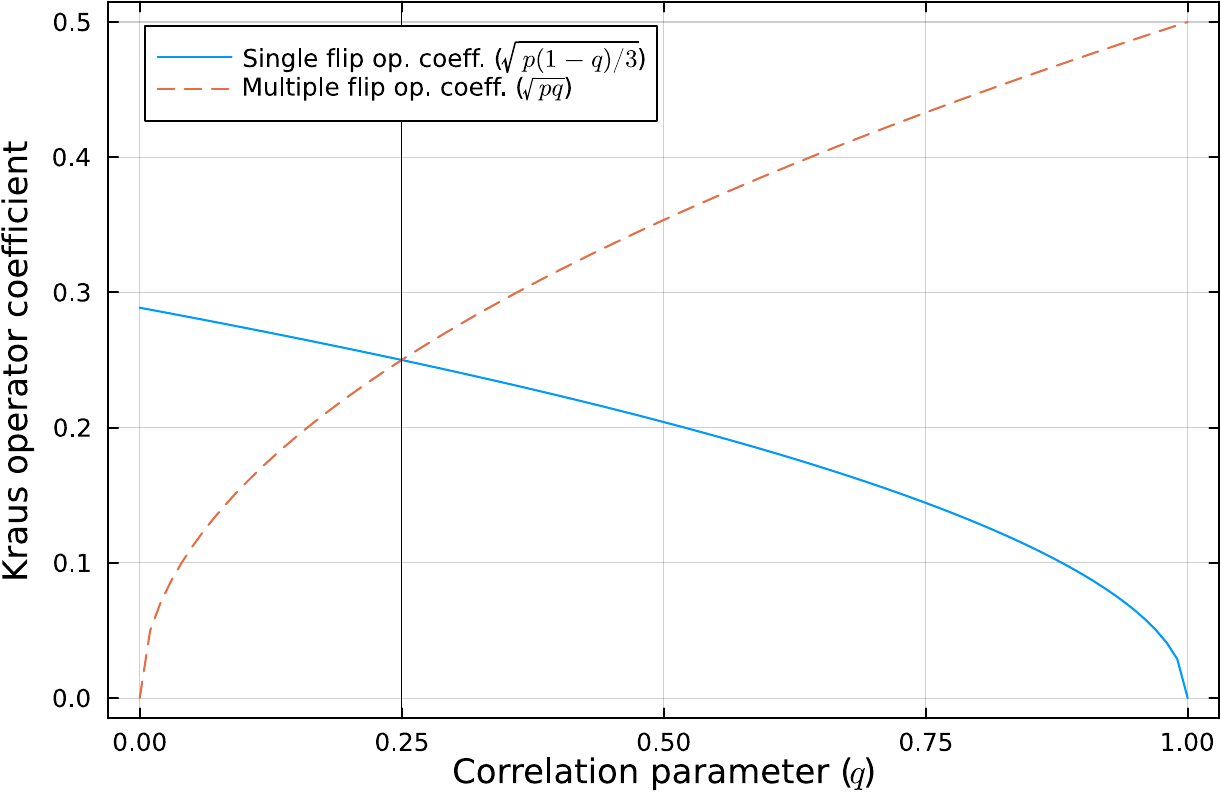}
    \caption{Coefficients of the Kraus operators corresponding to single bit-flips and a simultaneous bit-fip in the second and third qubits as described in \eqref{eq:toycorrbitflip}, for error probability $p=0.25$ and different values of the correlation parameter $q$. Notice how the crossing of these lines corresponds to the change of optimal code class observed in Figure \ref{fig:bitflip_corr}.}
    \label{fig:bitflip_corr_coeffs}
\end{figure}

\subsubsection{Full noise model on 3 qubits and approximate correction}

Finally, the full noise model, inclusive of simultaneous bit-flips on two or three different qubits, is the following
\begin{equation}
\label{eq:fullbitflip}
    \mathcal{N} \sim \{\bigotimes_{j=1}^3 (\sqrt{1-p} \one)^{1 - k_j}
    (\sqrt{p} \sigma_x)^{k_j} \ |
    \ \{k_j\}_{j=1}^3 \in \{0, 1\}\}.
\end{equation}
The optimization was performed for $p = 0.25$.
Notice that the independent terms remain dominant for $p < 0.5$. As such, following the
intuition from the toy model \eqref{eq:toycorrbitflip}, we would expect the optimized
codes to be optimal for upto $p = 0.5$. This intuition is consistent with the results
seen in Figure \ref{fig:bitflip}, as evidenced by the crossing at $p = 0.5$.
Notice also that all three codes have the exact same performance. This is in fact the case
regardless of the inclusion of the correlated terms. If these were excluded, i.e. if the code
was optimized for the noise model \eqref{eq:indybitflip} instead of \eqref{eq:fullbitflip}, the performance
would be the same.

When multiple correctable codes exist and the cost function is as regular as it is
for the bit-flip channel, regularization leads to simpler codes without losing performance.
In fact, we are able to obtain the repetition code, and equivalent codes with the same
sparsity, with $\lambda = 0.1$. Some examples of codewords are reported in  \ref{regularizationeffects}.

\begin{figure}[!htpb]
    \centering
    \includegraphics[width=0.9\columnwidth]{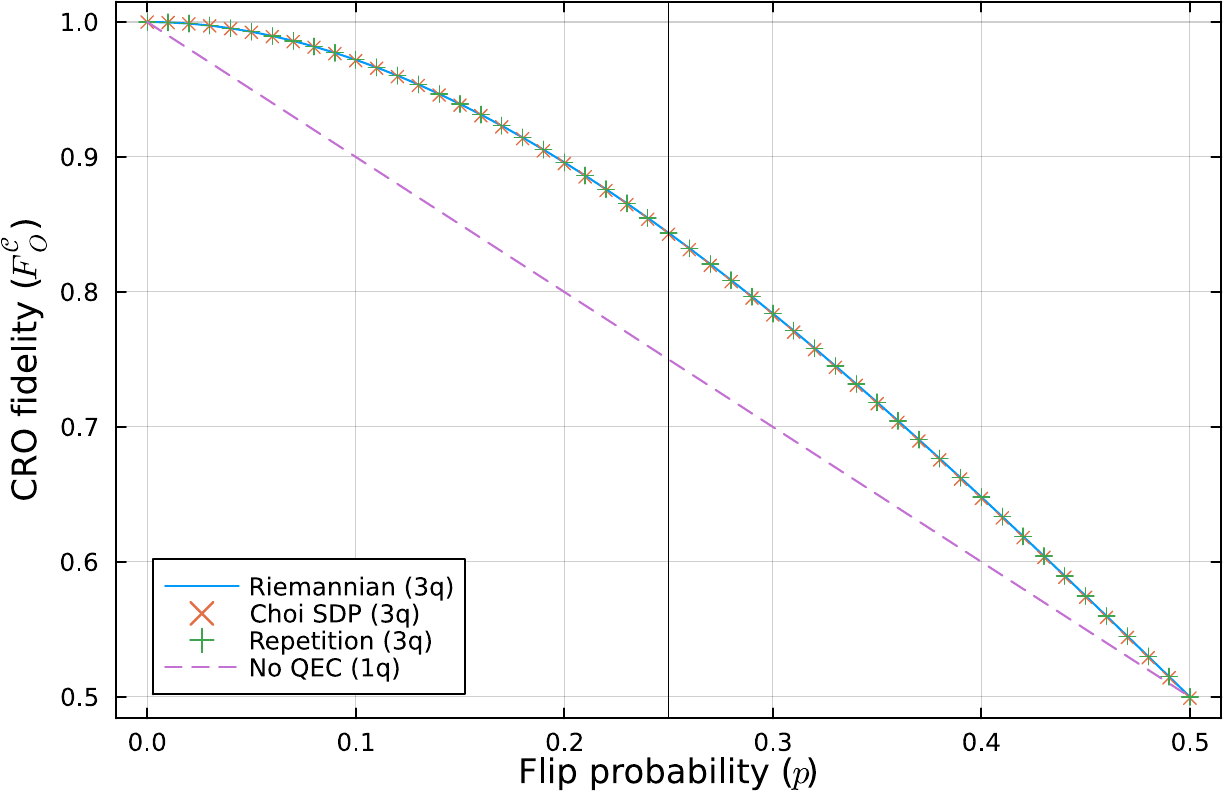}
    \caption{CRO fidelity of the corrected full bit-flip channel using Riemannian optimization (blue solid line), SDP (red x marks), the repetition code and recovery (green crosses), and with no correction (purple dashed line). The vertical line at $p = 0.25$ indicates the flip probability for which the optimization was performed. (3q) or (1q) indicates the number of qubits that compose the physical system. All of the tested QEC codes have the same performance.}
    \label{fig:bitflip}
\end{figure}

\begin{figure}[!htpb]
    \centering
    \includegraphics[width=0.9\columnwidth]{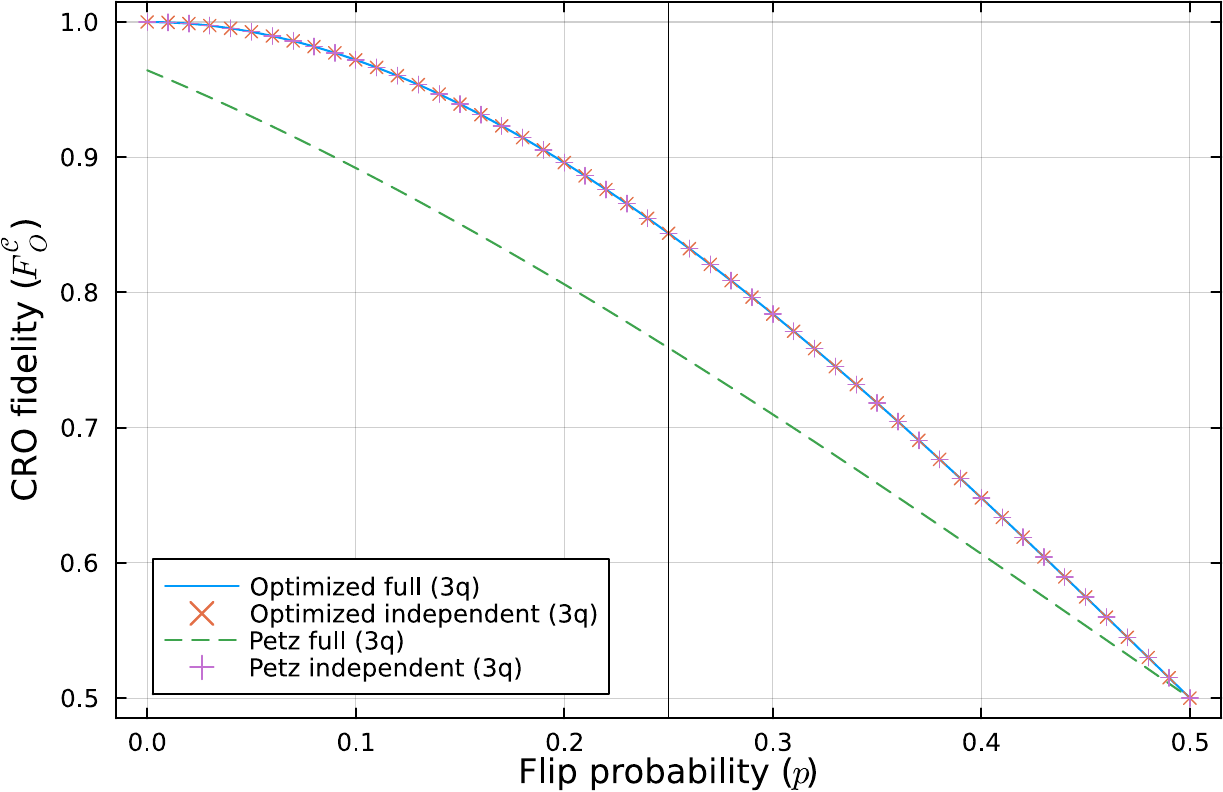}
    \caption{Comparison between the Petz recovery and optimized recoveries, tested on the full bit-flip channel with both the independent and full noise models. The Petz recovery is optimal for the independent noise model (purple crosses), as expected, since it is a perfectly correctable noise. Indeed, its performance overlaps with the optimized recovery (red x marks). For the full noise model, the Petz recovery is no longer optimal (green dashed line). However, it allows us to find the optimal subspace code, and then an optimal recovery (blue solid line) is found with the additional recovery optimization step. Notice how it overlaps with the recoveries designed for the independent noise model. This is in agreement with what was evidenced in Figure \ref{fig:bitflip_corr}. The vertical line at $p = 0.25$ indicates the flip probability for which the optimization was performed. (3q) refers to the number of qubits that compose the physical system.}
    \label{fig:bitflip_optim_recovery}
\end{figure}

\subsection{Optimal Codes for Amplitude Damping Noise}

The amplitude damping channel can be used to model the energy loss of the quantum system to its
environment, for instance an electron decaying from an excited to a ground state through spontaneous emission. Its Kraus operators are
given by the following two operators,
\begin{equation}
    E_0 = \begin{bmatrix}
        1 & 0 \\
        0 & \sqrt{1 - p}
    \end{bmatrix},
\quad
    E_1 = \begin{bmatrix}
        0 & \sqrt{p} \\
        0 & 0
    \end{bmatrix}.
\end{equation}
$E_1$ corresponds to an energy decay with probability $p$, whereas $E_0$ ensures
trace preservation.

No perfect code exist for this noise model, but an approximate four qubit code to protect a single encoded qubit has been proposed in
\cite{leungApproximateQuantumError1997}. Its explicit representation is given in equation \eqref{eq:leungcode}.
We next apply our method to the four qubit
amplitude damping channel. The Kraus
representation of the independent noise model reads
\begin{equation}
\label{eq:indyampdamp}
    \mathcal{N} \sim \{\sqrt{\frac{1}{4}} \bigotimes_{j=1}^4 E_0^{\delta_{jk}}\}_{k=1}^4 \cup
    \{\sqrt{\frac{1}{4}} \bigotimes_{j=1}^4 E_1^{\delta_{jk}}\}_{k=1}^4.
\end{equation}
Whereas the full noise model is given by
\begin{equation}
\label{eq:fullampdamp}
    \mathcal{N} \sim \{\bigotimes_{j=1}^4 E_0^{1 - k_j} E_1^{k_j} \ | \ \{k_j\}_{j=1}^4 \in \{0, 1\}^4\}.
\end{equation}

Figure \ref{fig:ampdamp} shows the performance of the codes and recoveries obtained with our method
and with the iterative SDP method proposed in \cite{lidarQuantumErrorCorrection2013}.
It also shows the difference in performance between using only the independent noise
for optimization or using the full noise model. We have also included
the approximate code from \cite{leungApproximateQuantumError1997}, and a single qubit without QEC.
In this case we do see a difference from optimizing with the full noise model or only with the independent model.
In fact, our method seems to be more robust to this change of noise model, giving a significantly better result
than the SDP method when the optimization is performed independent noise model, but used to correct the full noise.
It is also important to mention that our method was orders of magnitude faster (minutes vs hours), mainly due to the
fact that we optimize the code and recovery only once and that our optimization variables
of lower dimensionality. In addition, in optimizing the recovery operation
we can limit the Kraus rank of the recovery map, which we set to the rank of the noise model. On the other hand,
in the SDP method, one typically obtains full rank quantum channels both for encoding and recovery.

In approximate cases, regularization, albeit still useful
for simplifying the code, can be detrimental to fidelity optimization. A true
trade-off appears between the sparsity of the code and the fidelity, mediated by the
regularization hyperparameter $\lambda$.
Table \ref{tab:regularization} shows this effect. Notice the decrease
of fidelity for $\lambda = 0.005$. In this example we were able to reach the same level of
sparsity with $\lambda = 0.001$ and without any loss of fidelity with respect to the case
without regularization. The regularization parameter needs to be tuned for each noise model.

\begin{table}[!htpb]
\centering
\begin{tabular}{|c|c|c|c|}
\hline
$\lambda$    & 0.0    & 0.001  & 0.005  \\
\hline
$\ket{0000}$ & 0.6989 | 0.3466 & 0.6143 | 0.4811 & 0.5472 | 0.0   \\
$\ket{0001}$ & 0.0008 | 0.0006 & 0.0    | 0.0    & 0.0    | 0.0   \\
$\ket{0010}$ & 0.0005 | 0.0003 & 0.0    | 0.0    & 0.0    | 0.471 \\
$\ket{0011}$ & 0.2221 | 0.448  & 0.3083 | 0.3936 & 0.0    | 0.171 \\
$\ket{0100}$ & 0.0005 | 0.0003 & 0.0    | 0.0    & 0.0    | 0.502 \\
$\ket{0101}$ & 0.0    | 0.0    & 0.3083 | 0.3937 & 0.0    | 0.144 \\
$\ket{0110}$ & 0.2221 | 0.448  & 0.0    | 0.0    & 0.0    | 0.0   \\
$\ket{0111}$ & 0.0003 | 0.0002 & 0.0    | 0.0    & 0.4612 | 0.0   \\
$\ket{1000}$ & 0.0006 | 0.0008 & 0.0    | 0.0    & 0.0    | 0.515 \\
$\ket{1001}$ & 0.2221 | 0.448  & 0.0    | 0.0    & 0.0    | 0.0   \\
$\ket{1010}$ & 0.0    | 0.0    & 0.3083 | 0.3937 & 0.0994 | 0.0   \\
$\ket{1011}$ & 0.0005 | 0.0005 & 0.0    | 0.0    & 0.49   | 0.0   \\
$\ket{1100}$ & 0.2221 | 0.448  & 0.3083 | 0.3936 & 0.132  | 0.0   \\
$\ket{1101}$ & 0.0006 | 0.0006 & 0.0    | 0.0    & 0.4695 | 0.0   \\
$\ket{1110}$ & 0.0006 | 0.0003 & 0.0    | 0.0    & 0.0    | 0.0   \\
$\ket{1111}$ & 0.5605 | 0.2779 & 0.4924 | 0.3856 & 0.0    | 0.458 \\
\hline
Fid.         & 0.9034 & 0.9034 & 0.8886 \\
\hline
\end{tabular}
\caption{Absolute value of the coefficients of the base elements of codes optimized with different regularization parameters $\lambda$. In each column, the number to the left corresponds to the logical state $\ket{0_L}$, whereas the number to the right corresponds to the logical state $\ket{1_L}$.}
\label{tab:regularization}
\end{table}

\begin{figure}[!htpb]
    \centering
    \includegraphics[width=0.9\columnwidth]{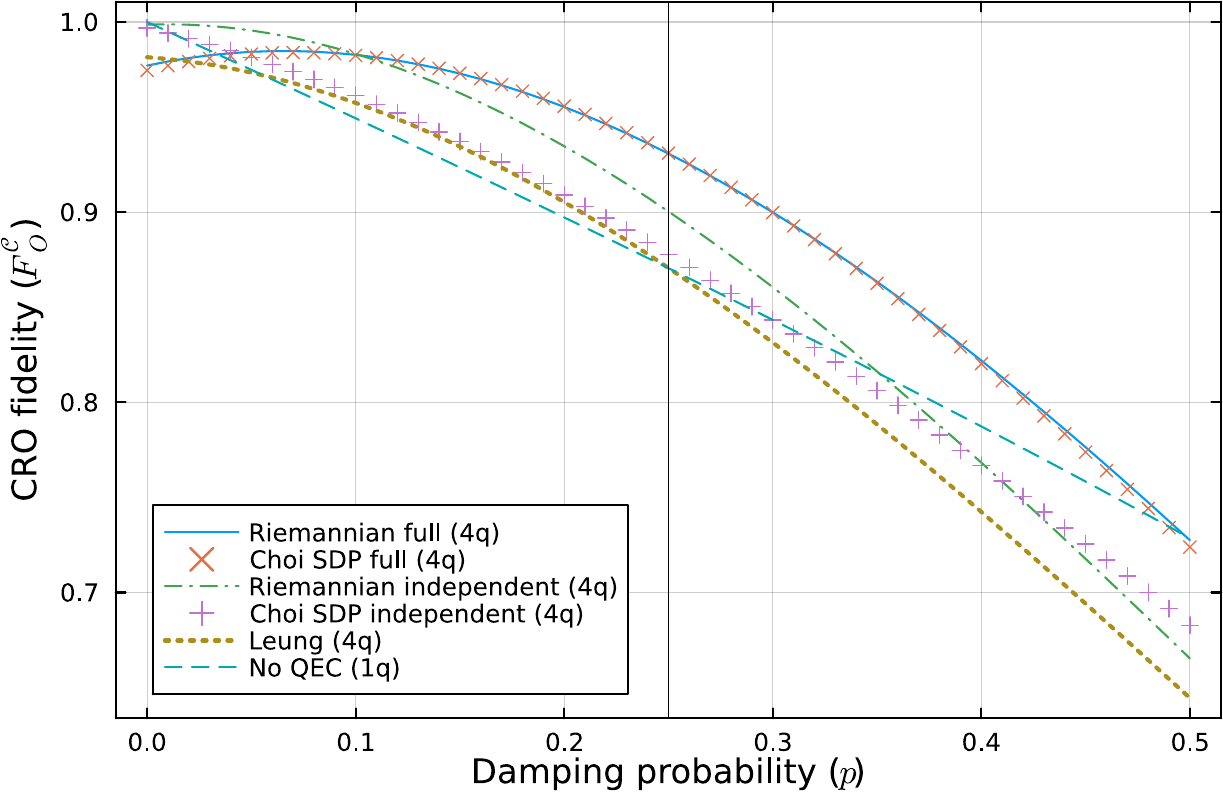}
    \caption{CRO fidelity of the corrected amplitude damping channel using Riemannian optimization (solid blue line for the full noise model and dot-dashed green line for the independent model), SDP (red x marks for the full noise model and purple crosses for the independent model), Leung's code and recovery (brown dotted line), and with no correction (cyan dashed line). The number between parentheses indicates the number of qubits that compose the physical system. The dashed line corresponds to the case of a single qubit without QEC. The vertical line at $p = 0.25$ indicates the damping probability for which the optimization was performed. The number between parentheses indicates the number of qubits that compose the physical system.}
    \label{fig:ampdamp}
\end{figure}

\begin{figure}[!htpb]
    \centering
    \includegraphics[width=0.9\columnwidth]{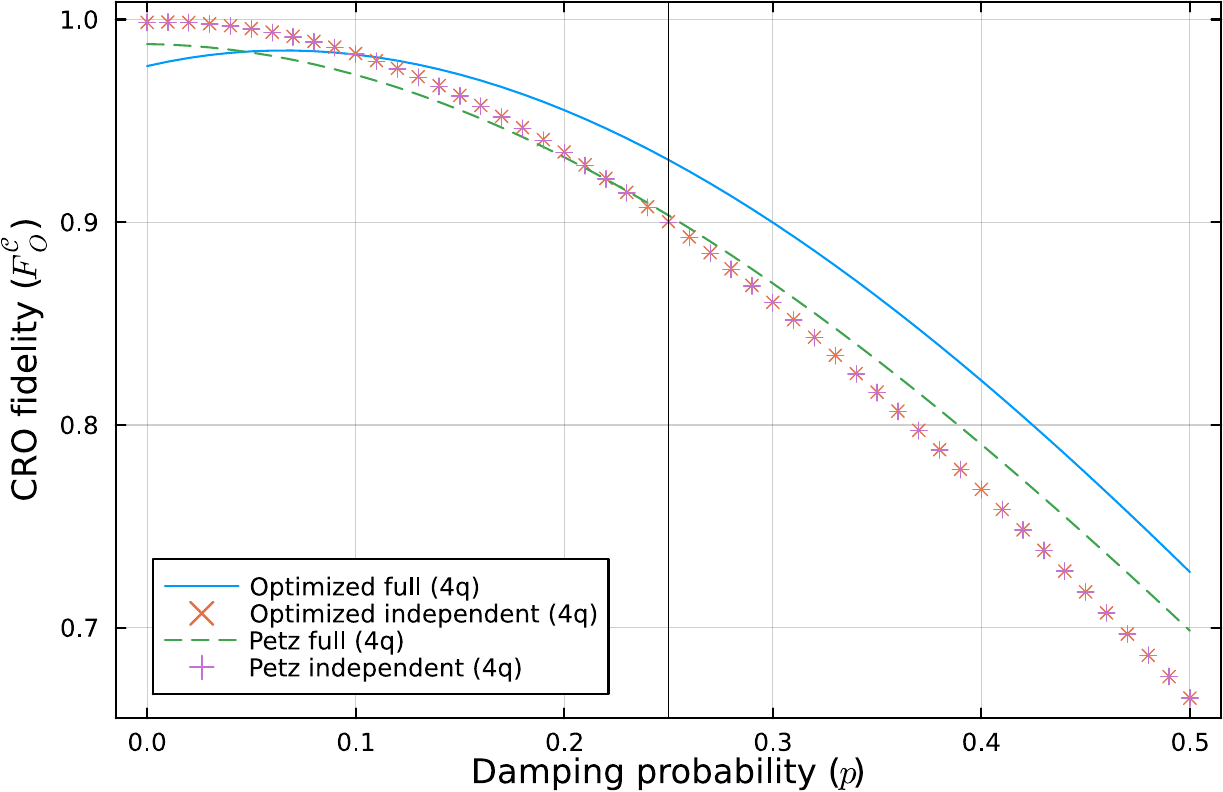}
    \caption{Comparison between the Petz recovery and optimized recoveries, tested on the full noise model but optimized with both the independent and the full models. As for the bit-flip noise, the Petz recovery is optimal for the independent noise model (purple crosses), notice how it overlaps with the optimized recovery (red x marks), but not for the full noise model (green dashed line), notice how it runs bellow the optimized recovery (blue solid line). Nevertheless, our approach still allows us to find an optimal subspace code without recurring to iterative optimization methods. The vertical line at $p = 0.25$ indicates the damping probability for which the optimization was performed. (4q) indicates the number of qubits that compose the physical system.}
    \label{fig:ampdamp_optim_petz}
\end{figure}

\subsection{Depolarizing noise on 5 qubits}

The depolarizing channel models isotropic loss of information in a qubit.
It can be obtained considering probabilistic flips around all three axes of the Bloch sphere, i.e.
bit-flip, phase-flip and bit-phase-flip. As mentioned before, the bit-flip is
represented by the first Pauli matrix $\sigma_x$. The other
two Pauli matrices,
\begin{equation}
 \sigma_z = \begin{bmatrix}
        1 & 0 \\
        0 & -1
    \end{bmatrix},    
\quad
   \sigma_y = \begin{bmatrix}
        0 & -\iu \\
        \iu & 0
    \end{bmatrix},
\end{equation}
induce a phase-flip and a bit-phase-flip, respectively. 

For a single qubit this channel has the following Kraus representation:
\begin{equation}
    \mathcal{N} \sim \{\sqrt{1 - p} \mathds{1},
    \sqrt{\frac{p}{3}} \sigma_x,
    \ \sqrt{\frac{p}{3}} \sigma_y,
    \ \sqrt{\frac{p}{3}} \sigma_z\}.
\end{equation}

The depolarizing channel represents general single qubit noise.
This is a consequence of the fact that if a code and recovery perfectly correct a noise with  Kraus operators $\{N_j\}$, then
they also correct any noise with Kraus operators
$\{\sum_k \lambda_{jk} N_k\}$. Since the Pauli matrices together with the identity
form a bases of $\mathbb{C}^{2 \times 2}$, a code able to correct
the depolarizing channel is also able to correct any other single qubit noise.

Different perfect five qubit code are known for this noise,
\cite{bennettMixedstateEntanglementQuantum1996, laflammePerfectQuantumError1996}.
These do not span the same subspace but are related to one another by local (single qubit) transformations.
The associated logical bases for these codes are reported in the Appendix, in equations \eqref{eq:perfectcodeBennett} and \eqref{eq:perfectcodeLaflamme},
respectively. 

For a system of five qubits, the following Kraus representation of the independent model is
\begin{equation}
\label{eq:indydepo}
    \begin{split}
        \mathcal{N} \sim \{\sqrt{1 - p} \mathds{1}^{\otimes 5}\} &\cup
        \{\sqrt{\frac{p}{15}} \bigotimes_{j=1}^5 \sigma_x^{\delta_{jk}}\}_{k=1}^5 \\ &\cup
        \{\sqrt{\frac{p}{15}} \bigotimes_{j=1}^5 \sigma_y^{\delta_{jk}}\}_{k=1}^5 \cup
        \{\sqrt{\frac{p}{15}} \bigotimes_{j=1}^5 \sigma_z^{\delta_{jk}}\}_{k=1}^5.
    \end{split}
\end{equation}
Whereas the full model is given by
\begin{equation}
\label{eq:fulldepo}
    \begin{split}
        \mathcal{N} \sim \{\bigotimes_{j=1}^5& (\sqrt{1-p} \one)^{\delta_{0,k_j}}
        (\sqrt{\frac{p}{3}} \sigma_x)^{\delta_{1,k_j}} \\
        & (\sqrt{\frac{p}{3}} \sigma_y)^{\delta_{2,k_j}}
        (\sqrt{\frac{p}{3}} \sigma_z)^{\delta_{3,k_j}} \ |
        \ \{k_j\}_{j=1}^5 \in \{0, 1, 2, 3\}^5\}.
    \end{split}
\end{equation}

Figure \ref{fig:depo} shows the performance of different codes and recoveries optimized for the
independent noise model \eqref{eq:indydepo} and then tested with the full noise model \eqref{eq:fulldepo}.
For comparison Laflamme's code and recovery \cite{laflammePerfectQuantumError1996} is included as well.

\begin{figure}[!htpb]
    \centering
    \includegraphics[width=0.9\columnwidth]{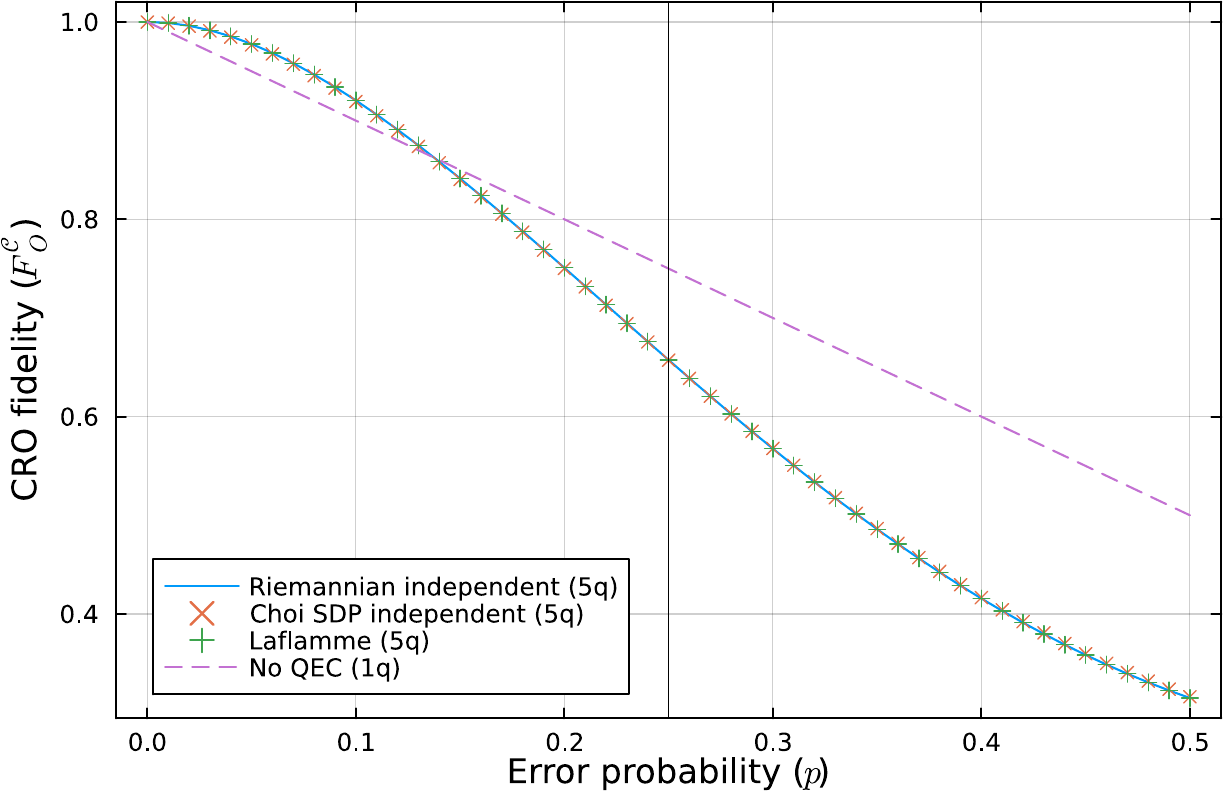}
    \caption{CRO fidelity of the corrected depolarizing channel using Riemannian optimization (blue solid line), SDP (red x marks), Laflamme's code and recovery (green crosses), and with no correction (purple dashed line). The vertical line at $p = 0.25$ indicates the error probability for which the optimization was performed. The number between parentheses indicates the number of qubits that compose the physical system. The codes were optimized for the independent noise model, but the fidelity shown was computed with the full noise model.}
    \label{fig:depo}
\end{figure}

\subsection{Non-Markovian pure dephasing}

A molecular nanomagnet in a spin bath has been
simulated by means of cluster-correlation expansion (CCE)
\cite{petiziolCounteractingDephasingMolecular2021, yangQuantumManybodyTheory2008, yangQuantumManybodyTheory2009},
then, the Kraus operators associated to the
dynamical map generated by the interaction with the bath have been
computed. This noise is of the pure dephasing kind, acting on a qudit.
Our optimization algorithm has been used to obtain the optimal quantum
codes.

Seeking an application to a noise model emerging from a physical Hamiltonian rather than a phenomenological description,
we considered a molecular nanomagnet $S$ of spin $7/2$ in a bath $B$ of $1/2$ spins.
This is represented by a Hilbert space $\mathcal{H}_S \otimes \mathcal{H}_B$ of dimension $8 \times 2^N$,
where $N$ is the amount of particles in the bath, which we set to $50$.
The effective
Hamiltonian for such a system has been taken from \cite{petiziolCounteractingDephasingMolecular2021}. First, CCE
\cite{yangQuantumManybodyTheory2008, yangQuantumManybodyTheory2009}
has been used to simulate the dephasing dynamics of the system. Then,
these simulations have been used to compute the Choi matrix, from which
the Kraus operators of the pure dephasing noise were computed.

The effective Hamiltonian in question is the following,
\begin{equation}
    H = \mathds{1}_S \otimes H_B^{(0)} + S_z \otimes H_B^{(1)} + (\hat{S}^2 - S_z^2) \otimes H_B^{(2)},
\end{equation}
where
\begin{equation}
    H_B^{(0)} = \sum_{n=1}^N (a_n^{(k)} I_n^z + b_n^{(k)} (I_n^z)^2) +
    \sum_{\substack{n,m=1 \\ m \neq n}}^N (c_{n,m}^{(k)} I_n^+ I_m^- + d_{n,m}^{(k)} I_n^z I_m^z),
\end{equation}
and $\{S^x, S^y, S^z\}$ are the spin operators of $S$,
$\hat{S}^2 = S_x^2 + S_y^2 + S_z^2$ is the effective spin operator squared and
$\{I_n^x, I_n^y, I_n^z\}$ are the spin operators of the $n$-th spin of $B$.

Let
$H_{B,m} = (\ketbra{m}{m} \otimes \mathds{1}_B) H (\ketbra{m}{m} \otimes \mathds{1}_B)$,
be the projection of the effective Hamiltonian onto each of the
eigenspaces of $S^z$, where $S^z \ket{m} = m \ket{m}$, and
$m \in \{-7/2, -5/2, -1/2, 1/2, 3/2, 7/2\}$.
It is easy to see that the effective Hamiltonian is block-diagonal and that
$H = \sum_{m = 0}^7 H_{B,m}$. Therefore, the bath state $\ket{\mathcal{J}}$ evolves
conditioned on $\ket{m}$  as
$\ket{\mathcal{J}^m(t)} = e^{-\iu H_{B,m} t} \ket{\mathcal{J}}$. Then, the joint
state of $S$ and $B$ evolves to the entangled state
$\ket{\psi(t)} = \sum_m c_m \ket{m} \otimes \ket{\mathcal{J}^m(t)}$.
Computing the reduced state $\rho_s(t) = \mathrm{tr}_B(\ketbra{\psi(t)}{\psi(t)})$,
we obtain that $\rho_s(t)$ has entries
$\rho_{m,n}(t) = \rho_{m,n}(0) \braket{\mathcal{J}^n(t)}{\mathcal{J}^m(t)}$.
Therefore, the dynamical map $\rho_s(0) \mapsto \rho_s(t)$ can be characterized by the set of
functions $L_{m,n}(t) = \braket{\mathcal{J}^n(t)}{\mathcal{J}^m(t)}$, where
$m,n \in \{-7/2, -5/2, -1/2, 1/2, 3/2, 7/2\}$, which can be approximated by CCE. Finally,
given $t$, the Choi matrix of this dynamical map can be computed and then used to obtain the
corresponding Kraus operators.

However, the resulting Kraus map cannot be iterated, because the map
$\rho_s(t_0 + k \Delta t) \mapsto \rho_s(t_0 + (k+1) \Delta t)$ is not Markovian. 
This means that
the error correcting procedure can only reliably be applied once using this model.
Keeping this in mind, Table \ref{tab:dephasing} shows the performance of the one-shot recoverability of the information, yielding excellent performance, and improving on randomized searches for codes.

\begin{table}[!htpb]
\centering
\begin{tabular}{|c|c|c|c|c|}
\hline
          & Optimal   & Repetition      & Rand. mean    & Rand. best   \\
\hline
    Fid.  & 0.9952    & 0.5006    & 0.8981        & 0.9744       \\
\hline
\end{tabular}
\caption{Performance of different codes under the pure dephasing noise.}
\label{tab:dephasing}
\end{table}

In order to study the performance degradation of further applications of the QEC procedure due to the non-Markovianity
of the noise, we consider a toy model with only two bath spins. This allows us to simulate the exact state evolution of
the system, including the bath. This, in turn, allows us apply the recovery map periodically throughout the system's
evolution. Figure \ref{fig:markov} shows the CRO fidelity after each application of the recovery map.

\begin{figure}[!htpb]
    \centering
    \includegraphics[width=0.9\columnwidth]{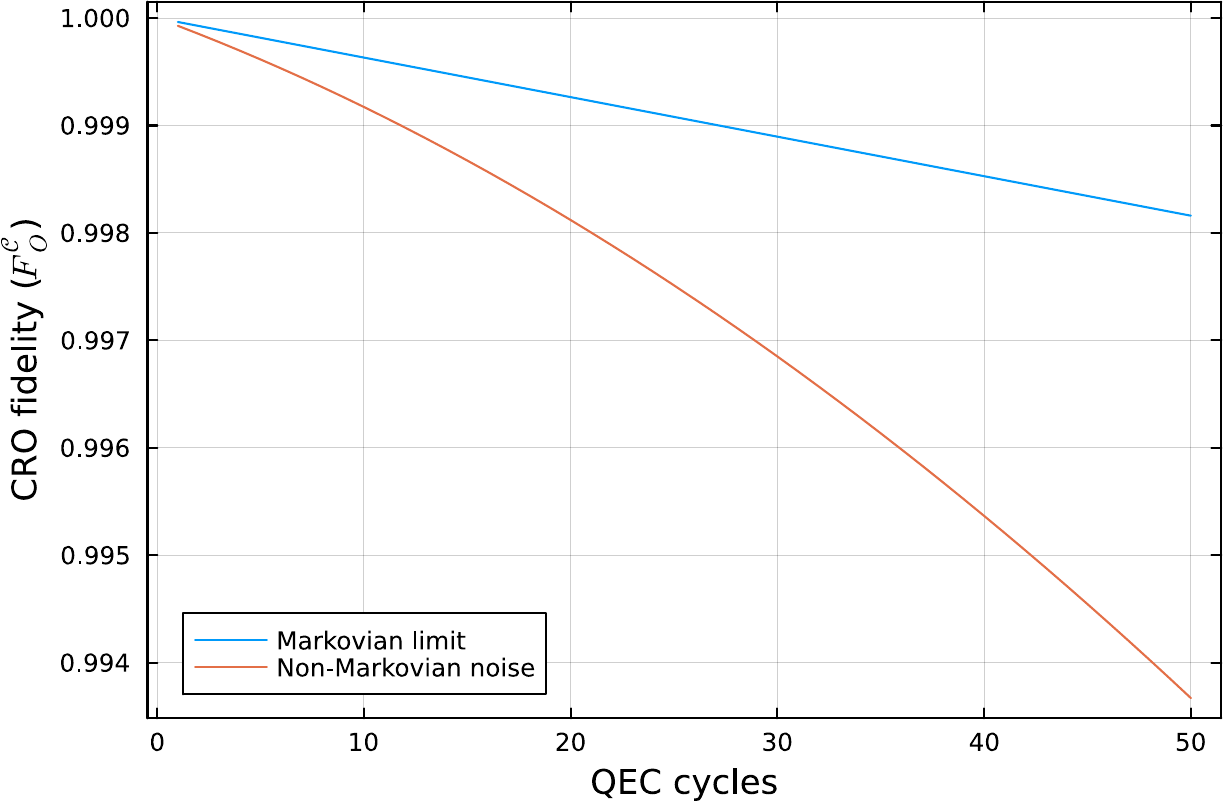}
    \caption{Performance degradation due to the non-Markovianity of the noise. The blue line shows the performance if the noise was Markovian and the red line shows the performance with the actual non-Markovian pure dephasing noise. In both cases the same subspace code is used, in combination with the same Petz recovery.}
    \label{fig:markov}
\end{figure}

\subsection{IBM qubits}

We next consider the noise acting on the first three qubits of IBM's \texttt{ibm\_kyoto} quantum computer.
Quantum process tomography has been performed on the first three
qubits of one of IBM's Eagle r3 quantum processors, in order to characterize
the noise that affects these qubits, and obtain the corresponding Kraus
operators. Next, our algorithm has been applied to the reconstructed noise model in order
to obtain an optimal quantum code to represent and correct a single qubit using these three as a physical support.

In order to obtain a model of the noise acting on the first three qubits of IBM Kyoto quantum computer,  quantum process tomography circuits were implemented using Qiskit \cite{Qiskit}.
The implementation was based on the tomography module of the Qiskit Experiments extension.
However, this extension was not used, since we needed greater control over the job submission
process. The input states used were
$\ket{\psi_0} = \ket{\psi_0^1} \otimes \ket{\psi_0^2} \otimes \ket{\psi_0^3}$, where
$\ket{\psi_0^k} \in \{\ket{0}, \ket{1}, \ket{+}, \ket{\iu}\}$, and the measurement
bases were $\sigma^1 \otimes \sigma^2 \otimes \sigma^3$, where
$\sigma^k \in \{\sigma_x, \sigma_y, \sigma_z\}$. Then experiment consisted in
preparing each of the input states, letting them freely evolve for 25 $\mu$s, and measuring
in each of the bases.
The time of 25 $\mu$s is close to a fourth of the median T2 time of the \texttt{ibm\_kyoto} processor, which is
110.48 $\mu$s, and to an eighth of its median T1 time, which is 212.5 $\mu$s.
Then, in order to fit a model to the
experiment's results, i.e. to find the Kraus operators of the noise,
the method developed in \cite{ahmedGradientDescentQuantumProcess2023}
was used for two main reasons. First, it allows us to reduce the amount
of parameters to be estimated by introducing an estimate of the Kraus rank of the
noise, and second, there is no need to project the estimate of the map afterward,
onto the space of CPTP maps.

The Kraus rank was estimated by looking at the convergence of the average square error
for each truncation of the map. Figure \ref{fig:krausrank} shows this convergence.
From this plot, the Kraus rank is estimated to be 12.

\begin{figure}[!htpb]
    \centering
    \includegraphics[width=0.9\columnwidth]{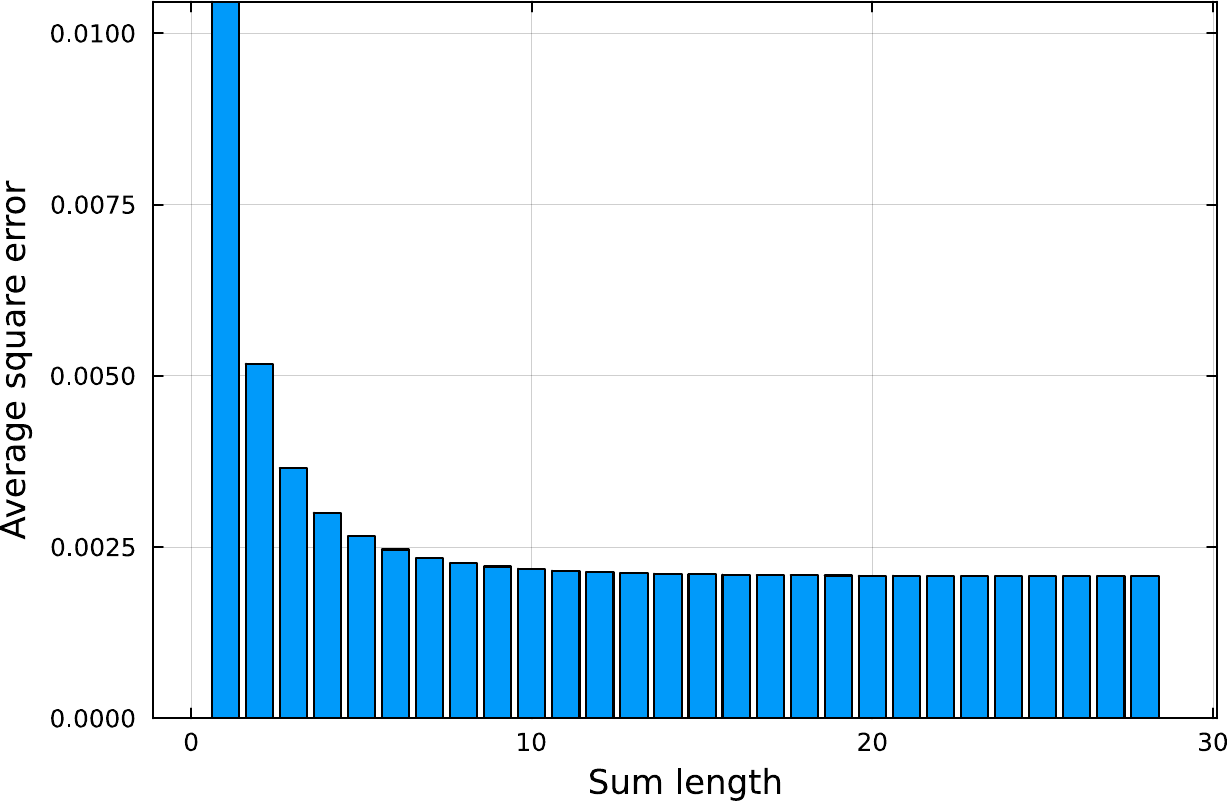}
    \caption{Average square error for the estimate of the noise of the first three qubits of IBM's \texttt{ibm\_kyoto} quantum computer in function of the length of the Kraus operator sum.}
    \label{fig:krausrank}
\end{figure}

Given the noise model estimate, our method has been used to find the best protected implementation of a single logical qubit in these three physical ones. The results are  reported in Table \ref{tab:ibm}, and show that the optimized code performance improve on either standard or randomized codes. The limited effectiveness of the error correction in this setting points to a noise process lacking intrinsic symmetries and a sampling time (25 $\mu$s) that is probably too long for the this system.

\begin{table}[!htpb]
\centering
\begin{tabular}{|c|c|c|c|c|}
\hline
          & Optimal   & Repetition      & Rand. mean    & Rand. best   \\
\hline
    Fid.  & 0.6053    & 0.5855    & 0.4697        & 0.5367       \\
\hline
\end{tabular}
\caption{Performance of different codes under IBM's qubits' noise.}
\label{tab:ibm}
\end{table}

\section{Conclusion}

A tune-able algorithmic solution to the problem of finding optimal subspace QEC
codes based on Riemannian optimization and QEC-theoretic consideration has been
proposed, and its effectiveness tested in some paradigmatic examples. The method has a
computational burden that is significantly lighter compared to existing iterative
SDP-based methods, guarantees the subspace structure of the code, and introduces for the
first time the possibility of selecting codes with simpler descriptions via regularization.

As illustrated in the numerical examples, our optimization scheme allows us to
find perfectly correctable codes when they exist, as in the case of the three
qubit bit-flip noise and the five qubit code for general uncorrelated errors.
When the noise is not correctable, our scheme allows us to find codes that
match or surpass the performance of known codes from the literature. In the
case of the four qubit amplitude damping noise, we are able to match the
performance of Leung's approximate code. Whereas in the more realistic cases
that we studied, i.e. the simulated non-markovian pure dephasing noise induced
by a spin bath, and the real noise in IBM's \texttt{ibm\_kyoto} quantum
computer, we were able to find codes that performed better than existing codes.
These benchmark examples prove the potential of our approach in finding QEC
codes with competitive performance in new scenarios, with noisy quantum
evolutions that have not yet been tamed by suitable QEC schemes.

Since it leverages the properties of the Stiefel manifold, our method has the
advantage that it does not require the implementation of barrier functions or
final projection steps, which in general degrade the code performance,
automatically returning  a valid quantum code and correction map at each run of
the algorithm. 

Of course, the natural development of this research will involve testing and
fine-tuning the algorithm in new real-world experimental settings, with noise
models coming from tomographic data (as in the IBM processor above) rather than
phenomenological or theoretical derivations. Extensions or refinements of the
computational steps are also possible. For instance, one could consider the
implementation of conjugate descent methods
\cite{satoRiemannianOptimizationIts2021, zhuRiemannianConjugateGradient2017}.

Given that we need to store the quantum code in a matrix,
the dimension of the physical system for which we can optimize codes is chiefly limited
by the amount of memory accessible to the computer. It is the common pitfall of any
method that needs to represent quantum states or quantum channels in terms of classical
information, as the dimension of the system grows exponentially with the amount of
qubits. On the other hand, the dimension of the code itself does not have such a great
impact in the feasibility of the algorithm. Indeed, the main constraint would be
the representation of the noise map, which is independent of the dimension of the code.

However, one would still want to be able to tackle large physical systems. If a system
is known to only have nearest neighbor interactions, one approach could be to design
local codes that protect against the local noise at different sections of the quantum
device. Alternatively, our method could be used to optimize the building blocks of a concatenated
quantum code. Tensor networks are a promising tool for describing many body quantum systems,
and they have already been used to generalize the concept of code concatenation to that
of ``quantum Lego'' \cite{caoQuantumLegoBuilding2022}. A future line of research would be to apply our
algorithm to the optimization of the building blocks of scalable codes.

Lastly, the ability to rapidly test different scenarios numerically also makes
the proposed algorithm a useful tool to probe new ideas of theoretical
interest. For example, the numerical investigations carried out for this work
point to two interesting observations: (i) {\em subspace codes are
(near-)optimal also in approximate error correction}: In our approach we
optimize over subspace codes, a choice that might be sub-optimal when only
approximately correctable codes are available.  A posteriori, we find that they
perform as good as the more general, even non-isometric codes coming from the
encoding maps obtained from the existing Semi-Definite Programming (SDP) based
methods \cite{lidarQuantumErrorCorrection2013}. It is worth noting that the
latter are typically subject to numerical errors and need to be normalized in
order to obtain physically-acceptable maps, leading to further errors; (ii)
{\em The optimization of the code is independent to that of the recovery map}:
While further iterations of code/recovery using our methods are possible, and
have been implemented, numerical results in all the considered examples show
that the performance of a single optimization cycle is not significantly
improved by repeated optimization.

New in-depth analyses are needed in order to assess whether such observation
are valid in general and can lead to new theoretical results in approximate QEC
and improve its applicability to realistic noise models.

\section*{Acknowledgments}

F. Ticozzi and M. Casanova were supported by the European Union through
NextGenerationEU, within the National Center for HPC, Big Data and
Quantum Computing under Projects CN00000013, CN 1, and Spoke
10. K. Ohki was supported by JSPS KAKENHI Grant Number JP23K26126.

\appendix

\section{Gradient Computations}
\label{chp:gradient}

\subsection{Riemannian Gradient}
In order to use gradient methods to obtain a local optima of our optimization
problem, we need to compute the gradient of the cost function.
First, let us define the gradient of $J_d(U)$ over $V_d(\mathbb{C}^n)$.
Since we consider a complex function $J_d : \mathbb{C}^{n \times d} \supset V_d(\mathbb{C}^n) \rightarrow [0, \infty)$,
the Euclidean gradient of $J_d$ becomes
\begin{equation}
    \nabla_U J_d(U) = \frac{\partial}{\partial U_{\rm Re}} J_d(U) + \iu \frac{\partial}{\partial U_{\rm Im}} J_d(U) \in \mathbb{C}^{n \times d},
\end{equation}
where $U_{\rm Re} = (U+\overline{U})/2$ and $U_{\rm Im} = (U -\overline{ U})/(2i)$ are the
real and imaginary matrices of $U$, respectively. The gradient at $U \in V_d(\mathbb{C}^n)$
is calculated as
\begin{equation}
    \mathrm{grad} J_d(U) = \nabla_U J_d(U) - U(\nabla_U J_d(U))^\dagger U.
    \label{eq:manigrad}
\end{equation}

It is easy to verify that $\mathrm{grad} J_d(U) \in \mathcal{T}_U V_d(\mathbb{C}^n)$.

\begin{remark}
    For complex cost functions, the complex derivative does not
    make sense in general. For example, let us consider $J(z) = \lvert z - \alpha \rvert^2$
    where $\alpha \in \mathbb{C}$. In this case, the complex derivative does not exist.
    However, for $z = x + \iu y$, the derivative over the real and imaginary axes exist
    and are calculated as
    \[
        \frac{\partial J(z)}{\partial x} = 2 (x - \mathrm{\rm Re}(\alpha)),
        \ \frac{\partial J(z)}{\partial y} = 2 (y - \mathrm{\rm Re}(\alpha)).
    \]
    Then,
    \[
        \nabla J(z) = 2 (z - \alpha),
    \]
    and this is the natural gradient for the cost function.
\end{remark}

Since $\mathrm{tr}(A \otimes B) = \mathrm{tr}(A) \mathrm{tr}(B)$ and $\mathrm{tr}(A) = \mathrm{tr}(A^\top)$ for any square
complex matrices $A$ and $B$, the cost function (\ref{eq:cost})
\begin{equation}
    \begin{split}
        J_d(U) &= \sum_{k,l=1}^m \mathrm{tr}\left((\Pi_U^\top \otimes \Pi_U)
        N_k^\top \mathcal{N}(\Pi_U)^{-\top/2} \right.
        \left. \overline{N_l}
        \otimes N_l^\dagger \mathcal{N}(\Pi_U)^{-1/2} N_k\right) \\
        &= \sum_{k,l=1}^m \mathrm{tr}\left(\Pi_U^\top N_k^\top \mathcal{N}(\Pi_U)^{-\top/2} \overline{N_l}\right)
        \mathrm{tr}\left(\Pi_U N_l^\dagger \mathcal{N}(\Pi_U)^{-1/2} N_k\right) \\
        &= \sum_{k,l=1}^m \mathrm{tr}\left(N_l^\dagger \mathcal{N}(\Pi_U)^{-1/2} N_k \Pi_U\right)
        \mathrm{tr}\left(\Pi_U N_l^\dagger \mathcal{N}(\Pi_U)^{-1/2} N_k\right) \\
        &= \sum_{k,l=1}^m \mathrm{tr}\left(\Pi_U N_l^\dagger \mathcal{N}(\Pi_U)^{-1/2} N_k\right)^2.
    \end{split}
\end{equation}

Now we need to obtain the Euclidean gradient of $J_d$ at any $U_0 \in V_d(\mathbb{C}^n)$.

\subsection{Useful gradients}

The following are gradients that are going to appear due to the chain rule, when
computing the gradient of the cost function,
\begin{equation}
    \nabla_U \mathrm{tr}(C_1(\Pi_{U_0}) \Pi_U) \rvert_{U=U_0},
\end{equation}
and
\begin{equation}
    \nabla_U \mathrm{tr}(C_2(\Pi_{U_0}) \mathcal{N}(\Pi_U)^{-1/2}) \rvert_{U=U_0},
\end{equation}
where $C_j(U_0) = C_j(U_0)^\dagger$.

\begin{lemma}
    \label{gradAux3}
    Let $A \in \mathbb{C}^{n \times n}$ be a given Hermitian
    matrix. then,
    \begin{equation}
        \nabla_U \mathrm{tr}(A \Pi_U) = 2 A U.
    \end{equation}
\end{lemma}
\begin{proof}
    Since $U = U_{\rm Re} + \iu U_{\rm Im}$,
    \begin{equation}
        \begin{split}
            \mathrm{tr}(A \Pi_U) &= \mathrm{tr}(U^\dagger A U) \\
            &= \mathrm{tr}((U_{\rm Re} - \iu U_{\rm Im})^\top A (U_{\rm Re} + \iu U_{\rm Im})) \\
            &= \mathrm{tr}(U_{\rm Re}^\top A U_{\rm Re}) - \iu \mathrm{tr}(U_{\rm Im}^\top A U_{\rm Re}) +
            \iu \mathrm{tr}(U_{\rm Re}^\top A U_{\rm Im}) + \mathrm{tr}(U_{\rm Im}^\top A U_{\rm Im}).
        \end{split}
    \end{equation}

    Therefore $\frac{\partial}{\partial U_{\rm Re}} \mathrm{tr}(A \Pi_U), \frac{\partial}{\partial U_{\rm Im}} \mathrm{tr}(A \Pi_U) \in \mathbb{R}^{n \times d}$ are
    \begin{equation}
        \frac{\partial}{\partial U_{\rm Re}} \mathrm{tr}(A \Pi_U) =
        (A + A^\top) U_{\rm Re} + \iu (A - A^\top) U_{\rm Im},
    \end{equation}
    \begin{equation}
        \frac{\partial}{\partial U_{\rm Im}} \mathrm{tr}(A \Pi_U) =
        (A + A^\top) U_{\rm Im} - \iu (A - A^\top) U_{\rm Re},
    \end{equation}
    and
    \begin{equation}
    \begin{split}
        \nabla_U \mathrm{tr}(A \Pi_U) &=
        \frac{\partial}{\partial U_{\rm Re}} \mathrm{tr}(A \Pi_U) +
        \iu \frac{\partial}{\partial U_{\rm Im}} \mathrm{tr}(A \Pi_U) \\
        &=
        2 A U.
    \end{split}
    \end{equation}
\end{proof}

\begin{lemma}
    \label{gradAux2}
    Let $A \in \mathbb{C}^{n \times n}$ be a given Hermitian matrix.
    Then,
    \begin{equation}
        \nabla_U \mathrm{tr}(A \mathcal{N}(\Pi_U)^{-1/2}) =
        2 \sum_{k=1}^m \sum_{j=1}^3 N_k^\dagger \mathrm{vec}^{-1}(\hat{C}^{-1} \mathrm{vec}(C_j)^\dagger N_k) U,
    \end{equation}
    where
    \begin{equation}
        C_1 = - \mathcal{N}(\Pi_{U_0})^{-1/2} A \mathcal{N}(\Pi_{U_0})^{-1/2},
    \end{equation}
    \begin{equation}
        C_2 = \mathcal{N}(\Pi_{U_0})^{-1} A (\mathds{1}_n - \mathcal{N}(\Pi_{U_0})^{-1/2} \mathcal{N}(\Pi_{U_0})^{1/2}),
    \end{equation}
    \begin{equation}
        C_3 = (\mathds{1}_n - \mathcal{N}(\Pi_{U_0})^{-1/2} \mathcal{N}(\Pi_{U_0})^{1/2}) A \mathcal{N}(\Pi_{U_0})^{-1},
    \end{equation}
    \begin{equation}
        \hat{C} = \mathcal{N}(\Pi_U)^{\top/2} \otimes \mathds{1}_n + \mathds{1}_n \otimes \mathcal{N}(\Pi_U)^{1/2}.
    \end{equation}
\end{lemma}
\begin{proof}
    For any non-zero matrix $X = (X_{jk}) \in \mathbb{C}^{n \times n}$ and
    $A \in \mathbb{C}^{n \times n}$, the derivative of the Moore-Penrose pseudo inverse
    is given as
    \begin{equation}
    \begin{split}
        \frac{\partial}{\partial X_{jk}} X^{-1} =&
        - X^{-1} \frac{\partial X}{\partial X_{jk}} X^{-1}
        + (\mathds{1}_n - X^{-1} X) \frac{\partial X^\dagger}{\partial X_{jk}} (X^{-1})^\dagger X^{-1} \\
        &+ X^{-1} (X^{-1})^\dagger \frac{\partial X^\dagger}{\partial X_{jk}} (\mathds{1}_n - X X^{-1}),
    \end{split}
    \end{equation}
    where the derivative is the Euclidean derivative
    \cite{golubDifferentiationPseudoInversesNonlinear1973, kolihaContinuityDifferentiabilityMoorePenrose2001}.
    Therefore, we only consider the derivative of the underlined parts of the
    following equation,
    \begin{equation}
        \begin{split}
            \nabla_U \mathrm{tr}(A \mathcal{N}(&\Pi_U)^{-1/2}) \rvert_{U=U_0}
            = - \nabla_U \mathrm{tr}(A \mathcal{N}(\Pi_{U_0})^{-1/2} \underline{\mathcal{N}(\Pi_U)^{1/2}}
            \mathcal{N}(\Pi_{U_0})^{-1/2}) \rvert_{U=U_0} \\
            & + \nabla_U \mathrm{tr}(A (\mathds{1}_n - \mathcal{N}(\Pi_{U_0})^{-1/2} \mathcal{N}(\Pi_{U_0})^{1/2})
            \underline{\mathcal{N}(\Pi_U)^{1/2}} \mathcal{N}(\Pi_{U_0})^{-1}) \rvert_{U=U_0} \\
            & + \nabla_U \mathrm{tr}(A \mathcal{N}(\Pi_{U_0})^{-1} \underline{\mathcal{N}(\Pi_U)^{1/2}}
            (\mathds{1}_n - \mathcal{N}(\Pi_{U_0})^{-1/2} \mathcal{N}(\Pi_{U_0})^{1/2})) \rvert_{U=U_0}.
        \end{split}
        \label{eq:biggrad}
    \end{equation}

    Note that if $\mathcal{N}(\Pi_U)$ is invertible, i.e., positive definite, the
    above equation becomes
    \begin{equation}
        \nabla_U \mathrm{tr}(A \mathcal{N}(\Pi_U)^{1/2}) \rvert_{U=U_0} =
        -\nabla_U \mathrm{tr}(A \mathcal{N}(\Pi_{U_0})^{-1/2} \mathcal{N}(\Pi_U)^{1/2}
        \mathcal{N}(\Pi_{U_0})^{-1/2}) \rvert_{U=U_0}.
    \end{equation}

    Next, we establish the derivative of $\mathcal{N}(\Pi_U)^{1/2}$ in \eqref{eq:biggrad}.
    Since $\mathcal{N}(\Pi_U) = \mathcal{N}(\Pi_U)^{1/2} \mathcal{N}(\Pi_U)^{1/2}$,
    we have the following Sylvester equation,
    \begin{equation}
        \frac{\partial}{\partial U_{jk}} \mathcal{N}(\Pi_U) =
        \frac{\partial \mathcal{N}(\Pi_U)^{1/2}}{\partial U_{jk}} \mathcal{N}(\Pi_U)^{1/2}
        + \mathcal{N}(\Pi_U)^{1/2} \frac{\partial \mathcal{N}^{1/2}}{\partial U_{jk}}.
    \end{equation}

    By using vectorization,
    \begin{equation}
    \begin{split}
        \mathrm{vec}(\frac{\partial}{\partial U_{jk}} \mathcal{N}(\Pi_U)) =&
        \hat{C}\mathrm{vec}(\frac{\partial \mathcal{N}(\Pi_U)^{1/2}}{\partial U_{jk}}),
    \end{split}
    \end{equation}
    where
    \[
        \hat{C} := \mathcal{N}(\Pi_U)^{\top/2} \otimes \mathds{1}_n +
        \mathds{1}_n \otimes \mathcal{N}(\Pi_U)^{1/2}.
    \]
    
    For $C \in \mathbb{C}^{n \times n}$ and $b \in \mathbb{C}^n$, the general
    solution of $C x = b$ is given by $x = C^{-1} b + (\mathds{1}_n - C^{-1} C)d$
    if the solution exists, where $d \in \mathbb{C}^n$ is an arbitrary vector
    (we consider the Moore-Penrose pseudo-inverse if $\mathrm{det}(C) = 0$). Hence,
    \begin{equation}
        \label{eq:derivative_sqrt_matrix}
        \frac{\partial \mathcal{N}(\Pi_U)^{1/2}}{\partial U_{jk}} =
        \mathrm{vec}^{-1}(\hat{C}^{-1} \mathrm{vec}(\frac{\partial}{\partial U_{jk}} \mathcal{N}(\Pi_U))
        + (\mathds{1}_{n^2} - \hat{C}^{-1} \hat{C}) z_{jk}),
    \end{equation}
    where $z_{jk} \in \mathbb{C}^{n^2}$ is arbitrarily chosen. Now we analyze
    the effect of $z_{jk}$.  
    Note that since $\mathcal{N}(\Pi _{U_{0}}) $ is Hermitian and non-negative, there exists a unitary matrix $V \in \mathbb{C}^{n\times n}$ such that
    \begin{align*}
        \mathcal{N}(\Pi _{U_{0}}) = V \begin{bmatrix}
            \Lambda & \\ & O_{p}
        \end{bmatrix}
        V^{\dagger}
        ,
    \end{align*}
    where $\Lambda \in \mathbb{R}^{(n-p) \times (n-p)}$ is a diagonal matrix with strictly positive diagonal element and $1\leq p < n$ is the rank of the kernel of $\mathcal{N}(\Pi _{U_{0}}) $
    (here, we do not consider $p=0$ or $p=n$).
    Hence, we have
    \begin{equation}
        \hat{C} = \overline{V} \otimes V
        \left(
        \begin{bmatrix}
            \sqrt{\Lambda } & \\ & O_{p}
        \end{bmatrix} \otimes \one_{n} +
        \one_{n} \otimes
        \begin{bmatrix}
            \sqrt{\Lambda} & \\ & O_{p}
        \end{bmatrix}
        \right)
        V^{\top} \otimes V^{\dagger}
    \end{equation}
    and
    \begin{align*}
        I_{n^2} - \hat{C}^{-1}\hat{C}
        =&
        \overline{V} \otimes V \times
        \begin{bmatrix}
            O_{n(n-p)} &
            \\
            & O_{n-p} &
            \\
            & & \one_{p} &
            \\
            & & & \ddots &
            \\
            & & &  & O_{n-p} &
            \\
            & & &  &  & \one_{p}
        \end{bmatrix}
        \times V^{\top} \otimes V^{\dagger}
        \\
        =&
        \overline{V} \otimes V \times
        \left(
        \begin{bmatrix}
            O_{n-p} & \\ & \one_{p}
        \end{bmatrix}
        \otimes
        \begin{bmatrix}
            O_{n-p} & \\ & \one_{p}
        \end{bmatrix}
        \right)
        \times V^{\top} \otimes V^{\dagger}
        . 
    \end{align*}
    From \eqref{eq:derivative_sqrt_matrix}, $\mathcal{N}(\Pi_{U}) ^{-1/2}$ is multiplied from, at least, one of the both sides of the derivative of $\mathcal{N}^{-1/2}$, so for any $X \in \mathbb{C}^{n\times n}$,
    \begin{align*}
        (\mathcal{N}(\Pi_{U}) ^{-1/2} \otimes X) (\one_{n^2} - \hat{C}^{-1}\hat{C} ) &= \\
        (X\otimes \mathcal{N}(\Pi_{U}) ^{-1/2} ) (\one_{n^2} - \hat{C}^{-1}\hat{C} ) &= 0
        . 
    \end{align*}
    Therefore, without loss of generality, we can put $z_{jk} = 0$.

    For $C = C^\dagger \in \mathbb{C}^{n \times n}$, by using $\mathrm{tr}(A^\dagger B) = \mathrm{vec}(A)^\dagger \mathrm{vec}(B)$
    and the fact that $\hat{C}^\dagger = \hat{C}$,
    \begin{equation}
        \begin{split}
            \nabla_U \mathrm{tr}(C \mathcal{N}(\Pi_U)^{1/2}) =&
            \nabla_U \mathrm{vec}(C)^\dagger (\hat{C}^{-1} \mathrm{vec}(\mathcal{N}(\Pi_U))) \\
            =& \nabla_U (\hat{C}^{-1} \mathrm{vec}(C))^\dagger \mathrm{vec}(\mathcal{N}(\Pi_U)) \\
            =& \nabla_U \mathrm{tr}(\mathrm{vec}^{-1}(\hat{C}^{-1} \mathrm{vec}(C))^\dagger \mathcal{N}(\Pi_U)) \\
            =& \nabla_U \sum_{k=1}^m \mathrm{tr}(N_k^\dagger \mathrm{vec}^{-1}(\hat{C}^{-1} \mathrm{vec}(C))^\dagger N_k \Pi_U) \\
            =& 2 \left(\sum_{k=1}^m N_k^\dagger \mathrm{vec}^{-1}(\hat{C}^{-1} \mathrm{vec}(C))^\dagger N_k \right) U.
        \end{split}
        \label{eq:tallgrad}
    \end{equation}
    From (\ref{eq:biggrad}) and (\ref{eq:tallgrad}), we obtain the Euclidean
    gradient of $\mathrm{tr}(A \mathcal{N}(\Pi_U))^{-1/2}$.
\end{proof}

\subsection{Gradient of the cost function for code optimization}

We can now proceed to compute the Euclidean gradient of the cost function
\begin{equation}
    \begin{split}
        \nabla_U J_d(U)& \rvert_{U = U_0} = \nabla_U \sum_{k,l=1}^m \mathrm{tr}\left(\Pi_U N_l^\dagger \mathcal{N}(\Pi_U)^{-1/2} N_k\right)^2 \rvert_{U = U_0} \\
        = \sum_{k,l=1}^m &\left( 2 \mathrm{tr}\left(\Pi_{U_0} N_l^\dagger \mathcal{N}(\Pi_{U_0})^{-1/2} N_k\right) \right.
        (\nabla_U \mathrm{tr}\left(N_l^\dagger \mathcal{N}(\Pi_{U_0})^{-1/2} N_k \Pi_U\right) \rvert_{U = U_0} + \\
        &\quad \left. \nabla_U \mathrm{tr}\left(N_k \Pi_{U_0} N_l^\dagger \mathcal{N}(\Pi_U)^{-1/2}\right) \rvert_{U = U_0}) \right).
    \end{split}
\end{equation}

By applying Lemma \ref{gradAux3}
\begin{equation}
    \begin{split}
        \nabla_U J_d(U) \rvert_{U = U_0} =
        \sum_{k,l=1}^m &\left( 2 \mathrm{tr}\left(\Pi_{U_0} N_l^\dagger \mathcal{N}(\Pi_{U_0})^{-1/2} N_k\right) \right.
        (2 N_l^\dagger \mathcal{N}(\Pi_{U_0})^{-1/2} N_k U + \\
        &\quad \left. \nabla_U \mathrm{tr}\left(N_k \Pi_{U_0} N_l^\dagger \mathcal{N}(\Pi_U)^{-1/2}\right) \rvert_{U = U_0}) \right).
    \end{split}
\end{equation}

Finally, by applying Lemma \ref{gradAux2}
\begin{equation}
    \begin{split}
        \nabla_U J_d(U) \rvert_{U = U_0} =
        4 \sum_{j,k,l=1}^m &\left( \mathrm{tr}\left(\Pi_{U_0} N_l^\dagger \mathcal{N}(\Pi_{U_0})^{-1/2} N_k\right) \right.
        (N_l^\dagger \mathcal{N}(\Pi_{U_0})^{-1/2} N_k U + \\
        &\quad \left. \sum_{x=1}^3 (N_j^\dagger \mathrm{vec}^{-1}(\hat{C}^{-1} \mathrm{vec}(C_x))^\dagger N_j) U) \right),
    \end{split}
\end{equation}
where
\begin{equation*}
    C_1 = - \mathcal{N}(\Pi_{U_0})^{-1/2} N_k \Pi_{U_0} N_l^\dagger \mathcal{N}(\Pi_{U_0})^{-1/2},
\end{equation*}
\begin{equation*}
    C_2 = \mathcal{N}(\Pi_{U_0})^{-1} N_k \Pi_{U_0} N_l^\dagger (\mathds{1}_n - \mathcal{N}(\Pi_{U_0})^{-1/2} \mathcal{N}(\Pi_{U_0})^{1/2}),
\end{equation*}
\begin{equation*}
    C_3 = (\mathds{1}_n - \mathcal{N}(\Pi_{U_0})^{-1/2} \mathcal{N}(\Pi_{U_0})^{1/2}) N_k \Pi_{U_0} N_l^\dagger \mathcal{N}(\Pi_{U_0})^{-1},
\end{equation*}
\begin{equation*}
    \hat{C} = \mathcal{N}(\Pi_U)^{\top/2} \otimes \mathds{1}_n + \mathds{1}_n \otimes \mathcal{N}(\Pi_U)^{1/2}.
\end{equation*}

\subsection{Gradient of the cost function for recovery optimization}

We shift now our focus to the Euclidian gradient of the cost function
\eqref{eq:cost_recovery} of problem \eqref{eq:optimprob2}. In order to compute
it, we divide it into blocks that can be computed independently of each other.
Each block corresponds to one of the Kraus operators $R_k$, denoted by
$\nabla_{R_k} J(\hat{R})$. Remember also that we need to compute the derivatives
over the real and imaginary parts separately. Therefore we have that
\begin{equation*}
\begin{split}
\nabla_{R_k} J(\hat{R}) =& \nabla_{R_k} \sum_{j l} | {\rm tr}(R_l N_j \Pi) |^2 \\
=& (\frac{\partial}{\partial R_{k_{\rm Re}}} + \iu \frac{\partial}{\partial R_{k_{\rm Im}}})
\sum_j  {\rm tr}((R_{k_{\rm Re}} + \iu R_{k_{\rm Im}}) N_j \Pi)
{\rm tr}(\Pi N_j^\dagger (R_{k_{\rm Re}}^\top - \iu R_{k_{\rm Im}}^\top)).
\end{split}
\end{equation*}
The matrix derivative over the real and imaginary parts of $R_k$ are
\begin{equation*}
\begin{split}
\frac{\partial}{\partial R_{k_{\rm Re}}} J(\hat{R}) = \sum_j&
\left({\rm tr}(R_{k_{\rm Re}} N_j \Pi)
\frac{\partial}{\partial R_{k_{\rm Re}}} {\rm tr}(\Pi N_j^\dagger R_{k_{\rm Re}}^\top) \right. \\
& + {\rm tr}(\Pi N_j^\dagger R_{k_{\rm Re}}^\top)
\frac{\partial}{\partial R_{k_{\rm Re}}} {\rm tr}(R_{k_{\rm Re}} N_j \Pi) \\
& - \iu {\rm tr}(\Pi N_j^\dagger R_{k_{\rm Im}}^\top)
\frac{\partial}{\partial R_{k_{\rm Re}}} {\rm tr}(R_{k_{\rm Re}} N_j \Pi) \\
& \left. + \iu {\rm tr}(R_{k_{\rm Im}} N_j \Pi)
\frac{\partial}{\partial R_{k_{\rm Re}}} {\rm tr}(\Pi N_j^\dagger R_{k_{\rm Re}}^\top) \right) \\
= \sum_j& \left({\rm tr}(R_k N_j \Pi) \Pi N_j^\dagger \right.
\left. + {\rm tr}(\Pi N_j^\dagger R_k^\dagger) \Pi^\top N_j \right),
\end{split}
\end{equation*}
\begin{equation*}
\begin{split}
\frac{\partial}{\partial R_{k_{\rm Im}}} J(\hat{R}) = \sum_j&
\left( \iu {\rm tr}(\Pi N_j^\dagger R_{k_{\rm Re}}^\top)
\frac{\partial}{\partial R_{k_{\rm Im}}} {\rm tr}(R_{k_{\rm Im}} N_j \Pi) \right. \\
& - \iu {\rm tr}(R_{k_{\rm Re}} N_j \Pi)
\frac{\partial}{\partial R_{k_{\rm Im}}} {\rm tr}(\Pi N_j^\dagger R_{k_{\rm Im}}^\top) \\
& + {\rm tr}(\Pi N_j^\dagger R_{k_{\rm Im}}^\top)
\frac{\partial}{\partial R_{k_{\rm Im}}} {\rm tr}(R_{k_{\rm Im}} N_j \Pi) \\
& + \left. {\rm tr}(R_{k_{\rm Im}} N_j \Pi)
\frac{\partial}{\partial R_{k_{\rm Im}}} {\rm tr}(\Pi N_j^\dagger R_{k_{\rm Im}}^\top) \right) \\
= \iu \sum_j& \left( {\rm tr}(\Pi N_j^\dagger R_k^\dagger) \Pi^\top N_j^\top \right.
\left. - {\rm tr}(R_k N_j \Pi) \Pi N_j^\dagger \right).
\end{split}
\end{equation*}
Finally, putting everything together we obtain
\begin{equation}
\nabla_{R_k} J(\hat{R}) = 2 \sum_j {\rm tr}(\Pi N_j) \Pi N_j^\dagger.
\end{equation}

\section{Optimized code-words}\label{regularizationeffects}

\subsection{Bit-flip}

The basis of the repetition code for the bit-flip noise is
\begin{equation}
\label{eq:shorcode}
    \{\ket{0_L}, \ket{1_L}\} = \{\ket{000}, \ket{111}\}.
\end{equation}

The basis of the optimised code for this case is
\begin{equation}
    \begin{split}
        \{\ket{0_L}, \ket{1_L}\} = \{
            \ket{111}, \ket{000}
        \}.
    \end{split}
\end{equation}

\subsection{Amplitude damping}

The basis of Leung's code for the amplitude damping noise is
\begin{equation}
\label{eq:leungcode}
    \begin{split}
        \{\ket{0_L}, \ket{1_L}\} =
        \{& \frac{1}{\sqrt{2}} (\ket{0000} + \ket{1111}), \\
        & \frac{1}{\sqrt{2}} (\ket{0011} + \ket{1100})\}.
    \end{split}
\end{equation}

The following are the basis elements of an optimized code without regularization:

\begin{equation}
\begin{split}
\ket{0_L} =
& (-0.6158-0.3305i) \ket{00000} + \\
& (-0.0008-0.0001i) \ket{00001} + \\
& (0.0005+0.0002i) \ket{00010} + \\
& (-0.1407+0.1719i) \ket{00011} + \\
& (-0.0004-0.0003i) \ket{00100} + \\
& (0.1407-0.1719i) \ket{00110} + \\
& (-0.0003-0.0002i) \ket{00111} + \\
& (0.0006-0.0001i) \ket{01000} + \\
& (-0.1407+0.1719i) \ket{01001} + \\
& (0.0005+0.0i)    \ket{01011} + \\
&(-0.1407+0.1719i) \ket{01100} + \\
& (-0.0006+0.0i)    \ket{01101} + \\
& (0.0006+0.0003i) \ket{01110} + \\
& (-0.4948-0.2633i) \ket{01111}
\end{split}
\end{equation}

\begin{equation}
\begin{split}
\ket{1_L} =
& (-0.0732+0.3387i) \ket{00000} + \\
& (-0.0006+0.0001i) \ket{00001} + \\
&  (0.0002-0.0002i) \ket{00010} + \\
& (-0.4115-0.177i)  \ket{00011} + \\
&    (-0.0+0.0003i) \ket{00100} + \\
&  (0.4115+0.177i)  \ket{00110} + \\
&     (0.0+0.0002i) \ket{00111} + \\
&  (0.0008+0.0001i) \ket{01000} + \\
& (-0.4115-0.177i)  \ket{01001} + \\
&  (0.0005-0.0i)    \ket{01011} + \\
& (-0.4115-0.177i)  \ket{01100} + \\
& (-0.0006-0.0i)    \ket{01101} + \\
&  (0.0001-0.0003i) \ket{01110} + \\
& (-0.0578+0.2719i) \ket{01111}
\end{split}
\end{equation}

The following are the basis elements of an optimized code with regularization parameter $\lambda = 10^{-3}$:

\begin{equation}
\begin{split}
\ket{0_L} =
& (-0.343-0.5096i) \ket{00000} + \\
& (-0.0542+0.3035i) \ket{00011} + \\
& (-0.0542+0.3035i) \ket{00101} + \\
& (0.0542-0.3035i) \ket{01010} + \\
& (-0.0542+0.3035i) \ket{01100} + \\
& (-0.2748-0.4086i) \ket{01111}
\end{split}
\end{equation}

\begin{equation}
\begin{split}
\ket{1_L} =
& (0.4448+0.1834i) \ket{00000} + \\
& (0.1572+0.3609i) \ket{00011} + \\
& (0.1572+0.3609i) \ket{00101} + \\
& (-0.1572-0.3609i) \ket{01010} + \\
& (0.1571+0.3609i) \ket{01100} + \\
& (0.3565+0.1471i) \ket{01111}
\end{split}
\end{equation}

\subsection{Depolarizing}

The basis of Bennett's code is
\begin{equation}
\label{eq:perfectcodeBennett}
    \begin{split}
        \{\ket{0_L}, \ket{1_L}\} =
        \{&
        \frac{1}{4} (\ket{00000} - \ket{00011} \\ & + \ket{00101} - \ket{00110} \\
        & + \ket{01001} + \ket{01010} \\ & - \ket{01100} - \ket{01111} \\
        & - \ket{10001} + \ket{10010} \\ & + \ket{10100} - \ket{10111} \\
        & -\ket{11000} - \ket{11011} \\ & -\ket{11101} - \ket{11110}), \\
        &\frac{1}{4}(-\ket{00001} - \ket{00010} \\ & - \ket{00100} - \ket{00111} \\
        & - \ket{01000} + \ket{01011} \\ & + \ket{01101} - \ket{01110} \\
        & - \ket{10000} - \ket{10011} \\ & + \ket{10101} + \ket{10110} \\
        & - \ket{11001} + \ket{11010} \\ & - \ket{11100} + \ket{11111})
        \}.
    \end{split}
\end{equation}

The basis of Laflamme's code is
\begin{equation}
\label{eq:perfectcodeLaflamme}
    \begin{split}
        \{\ket{0_L}, \ket{1_L}\} =
        \{&
        \frac{1}{2 \sqrt{2}} (\ket{00000} - \ket{00110} \\
                              & - \ket{01001} \ket{01111} \\
                              & + \ket{10011} - \ket{10101} \\
                              & - \ket{11010} \ket{11100}
                              ), \\
        &\frac{1}{2 \sqrt{2}}(\ket{00011} + \ket{00101} \\
                              & + \ket{01010} - \ket{01100} \\
                              & - \ket{10000} \ket{10110} \\
                              & - \ket{11001} + \ket{11111}
                              ) \}.
    \end{split}
\end{equation}

The basis elements of the optimised code are
\begin{equation}
    \begin{split}
        \ket{0_L} =
            ( 0.3285 + 0.1308i) & \ket{00001} + \\
            (-0.2731 - 0.2246i) & \ket{00111} + \\
            ( 0.1634 + 0.3135i) & \ket{01000} + \\
            ( 0.0602 + 0.3484i) & \ket{01110} + \\
            ( 0.3451 - 0.0768i) & \ket{10010} + \\
            (-0.2197 - 0.2770i) & \ket{10100} + \\
            ( 0.2886 + 0.2043i) & \ket{11011} + \\
            (-0.0579 + 0.3488i) & \ket{11101},
    \end{split}
\end{equation}
\begin{equation}
    \begin{split}
        \ket{1_L} =
            (-0.1787 + 0.3050i) & \ket{00010} + \\
            (-0.3525 - 0.0277i) & \ket{00100} + \\
            ( 0.3461 - 0.0723i) & \ket{01011} + \\
            (-0.2160 - 0.2799i) & \ket{01101} + \\
            ( 0.2692 - 0.2292i) & \ket{10001} + \\
            ( 0.3262 - 0.1364i) & \ket{10111} + \\
            (-0.3535 - 0.0022i) & \ket{11000} + \\
            ( 0.3361 + 0.1097i) & \ket{11110}.
    \end{split}
\end{equation}

\subsection{Pure dephasing}

\begin{equation}
    \begin{split}
        \ket{0_L} =
            (-0.4765 - 0.0846i) & \ket{000} + \\
            (-0.5187 + 0.0279i) & \ket{010} + \\
            (-0.1554 + 0.5252i) & \ket{100} + \\
            (-0.4064 - 0.1755i) & \ket{110},
    \end{split}
\end{equation}
\begin{equation}
    \begin{split}
        \ket{1_L} =
            (-0.6126 - 0.0680i) & \ket{001} + \\
            ( 0.4551 - 0.2127i) & \ket{011} + \\
            (-0.5523 + 0.0666i) & \ket{101} + \\
            (-0.1791 + 0.1618i) & \ket{111}.
    \end{split}
\end{equation}

\subsection{IBM qubits}

\begin{equation}
    \begin{split}
        \ket{0_L} =
            (-0.2042 + 0.0287i) & \ket{000} + \\
            ( 0.2455 + 0.1646i) & \ket{001} + \\
            (-0.2760 + 0.1608i) & \ket{010} + \\
            ( 0.0405 + 0.1622i) & \ket{011} + \\
            (-0.2451 - 0.0723i) & \ket{100} + \\
            ( 0.7038 + 0.0386i) & \ket{101} + \\
            ( 0.2741 - 0.2389i) & \ket{110} + \\
            ( 0.1428 - 0.1593i) & \ket{111},
    \end{split}
\end{equation}
\begin{equation}
    \begin{split}
        \ket{1_L} =
            ( 0.0539 - 0.0575i) & \ket{000} + \\
            (-0.2920 + 0.0457i) & \ket{001} + \\
            ( 0.0144 + 0.5368i) & \ket{010} + \\
            ( 0.0589 - 0.1927i) & \ket{011} + \\
            ( 0.1207 + 0.1157i) & \ket{100} + \\
            (-0.0634 + 0.2779i) & \ket{101} + \\
            ( 0.2549 - 0.4780i) & \ket{110} + \\
            (-0.3117 + 0.2786i) & \ket{111}.
    \end{split}
\end{equation}

\clearpage

\bibliographystyle{iopart-num}
\bibliography{refs}

\end{document}